\documentclass[preprint,12pt]{elsarticle}
\usepackage{geometry}
\usepackage[hyphens]{url}
\usepackage{mypackage}

\journal{Computers and Electrical Engineering}

\begin{document}

\begin{frontmatter}



\title{Efficient reformulations of ReLU deep neural networks for surrogate modelling in power system optimisation}


\author[1]{Yogesh Pipada Sunil Kumar} 
\author[1]{S. Ali Pourmousavi}
\author[2]{Jon A.R. Liisberg}
\author[2]{Julian Lesmos-Vinasco}

\affiliation[1]{organization={School of Electrical and Mechanical Engineering, The University of Adelaide},
            addressline={Ingkarni Wardli Building}, 
            city={Adelaide},
            state={SA 5005},
            country={Australia}}
\affiliation[2]{organization={Watts A/S},
            addressline={Brogade 19D}, 
            city={K\o ge 4600},
            country={Denmark}}

\begin{abstract}
\noindent The ongoing decarbonisation of power systems is driving an increasing reliance on distributed energy resources, which introduces complex and nonlinear interactions that are difficult to capture in conventional optimisation models. As a result, machine learning based surrogate modelling has emerged as a promising approach, but integrating machine learning models such as ReLU deep neural networks (DNNs) directly into optimisation often results in nonconvex and computationally intractable formulations. This paper proposes a linear programming (LP) reformulation for a class of convexified ReLU DNNs with non-negative weight matrices beyond the first layer, enabling a tight and tractable embedding of learned surrogate models in optimisation. We evaluate the method using a case study on learning the prosumer's responsiveness within an aggregator bidding problem in the Danish tertiary capacity market. The proposed reformulation is benchmarked against state-of-the-art alternatives, including piecewise linearisation (PWL), MIP-based embedding, and other LP relaxations. Across multiple neural network architectures and market scenarios, the convexified ReLU DNN achieves solution quality comparable to PWL and MIP-based reformulations while significantly improving computational performance and preserving model fidelity, unlike penalty-based reformulations. The results demonstrate that convexified ReLU DNNs offer a scalable and reliable methodology for integrating learned surrogate models in optimisation, with applicability to a wide range of emerging power system applications.
\end{abstract}



\begin{keyword}
constraint learning,  neural networks, surrogate models, electricity market,  flexibility aggregator, demand response
\end{keyword}

\end{frontmatter}



\section*{List of Abbreviations}
\allowdisplaybreaks
\begin{flushleft}
\begin{tabular}{ll}
  cvxd & convexified \\
  DNN & Deep neural network \\
  LP & Linear program \\
  ML & Machine learning \\
  mFRR & Manual frequency restoration reserve \\ 
  MIP & Mixed integer program \\ 
  OCL & Optimisation with constraint learning \\
  PCAR & Penalised convex area relaxation \\
  PCTAR & Penalised convex triangle area relaxation \\
  PWL & Piecewise linearisation \\
  ReLU & Rectified linear unit\\
  RMSE & Root mean square error \\
  SOTA & State-of-the-art \\
  UC & Unconstrained 
\end{tabular}
\end{flushleft}
\nomenclature[A]{$\mathbb{R}$}{Set of real numbers}
\nomenclature[A]{$\mathbb{R}_{+}$}{Set of non-negative real numbers}
\nomenclature[A]{$\mathcal{T}$}{Set of time steps, $\{1, \dots , T\}$}
\nomenclature[A]{$\mathcal{L}$}{Index of ReLu DNN layers, $\{1, \dots , L+1\}$}
\nomenclature[A]{$\mathbb{N}$}{Set of natural numbers}
\nomenclature[B]{$T$}{Optimisation horizon, $T \in \mathbb{N}$}
\nomenclature[B]{$L$}{Number of hidden ReLu layers, $L\in \mathbb{N}$}
\nomenclature[B]{$\boldsymbol{A}$}{Cross-temporal matrices, $\boldsymbol{A} \in [0, 1]^{T \times T}$}
\nomenclature[B]{$\boldsymbol{\alpha_{l}}$}{Penalty coefficients for hidden layer $l\in \mathcal{L}$, $\boldsymbol{\alpha_{l}} \in \mathbb{R}^{n_{l}}$}
\nomenclature[B]{$\boldsymbol{W_{l}}$}{Weight matrix of hidden layer $l\in \mathcal{L}$, $\boldsymbol{W_{l}} \in \mathbb{R}^{n_{l-1} \times n_{l}}$}
\nomenclature[B]{$n_{l}$}{Number of neurons in hidden layer $l \in \mathcal{L}$, $n_{l} \in \mathbb{N}$}
\nomenclature[B]{$\boldsymbol{b_{l}}$}{Bias matrix of hidden layer $l \in \mathcal{L}$, $\boldsymbol{b_{l}} \in \mathbb{R}^{n_{l}}$}
\nomenclature[B]{$\boldsymbol{q}, \boldsymbol{r}$}{Prosumer responsiveness shaping parameters, $\boldsymbol{q}, \boldsymbol{r} \in \mathbb{R}_{+}^{T}$}
\nomenclature[B]{$\boldsymbol{\lambda^{\text{M}}}$}{Day-ahead capacity market prices in DKK/MWh, $\boldsymbol{\lambda^{\text{M}}} \in \mathbb{R}^{T}$}
\nomenclature[B]{$\boldsymbol{\overline{x}}$}{Maximum available flexibility in MWh, $\boldsymbol{\overline{x}} \in \mathbb{R}_{+}^{T}$}
\nomenclature[B]{$N_{P}$}{Number of pieces for PWL , $N_{P} \in \mathbb{N}$}
\nomenclature[B]{$N_{D}$}{Number of training data points, $N_{D} \in \mathbb{N}$}
\nomenclature[B]{$LB, UB$}{Bounds for PCTAR reformulation, $LB, UB \in \mathbb{R}$}
\nomenclature[C]{$\boldsymbol{x}$}{Flexibility volume offered to market in MWh, $\boldsymbol{x} \in \mathbb{R}_{+}^{T}$}
\nomenclature[C]{$\boldsymbol{\tilde{x}}$}{Updated maximum available flexibility in MWh, $\boldsymbol{\tilde{x}} \in \mathbb{R}_{+}^{T}$}
\nomenclature[C]{$\boldsymbol{h_{l}}$}{Output vector of hidden layer $l \in \mathcal{L}$, $h_{l} \in \mathbb{R}^{n_{l}}$}
\nomenclature[C]{$\boldsymbol{\lambda^{\text{I}}}$}{Per-unit price paid to prosumers for flexibility in DKK/MWh, $\boldsymbol{\lambda^{\text{I}}} \in \mathbb{R}^{T}$}
\nomenclature[C]{$\boldsymbol{\lambda^{\text{P}}}$}{Total flexibility purchase cost paid by aggregator in DKK, $\boldsymbol{\lambda^{\text{P}}} \in \mathbb{R}_{+}^{T}$}
\printnomenclature[1.7cm]
\section{Introduction}
\label{sec: Introduction}
\subsection{Motivation and related work}
\label{ssec: motivation}
The global transition to low-carbon energy systems is reshaping the technical and operational landscape of modern power grids. Three key developments drive this transition: the increasing penetration of renewable energy sources, the shift toward active monitoring and control within distribution networks, and the integration of emerging energy storage technologies such as batteries and hydrogen~\cite{Fernandez2024NetZeroReview}. While these trends are essential for achieving net-zero targets, they also introduce significant complexity into traditional power system optimisation and control. The resulting operational uncertainty, nonlinear and nonconvex interactions, and high-dimensional decision spaces often exceed the modelling capabilities of classical optimisation formulations. As a result, there is growing interest in combining machine learning (ML) with optimisation to better capture such complex behaviours.

Within this context, the optimisation-constrained learning (OCL) framework has emerged as a promising approach. OCL employs ML-based surrogate models to represent objective or constraint components that are analytically intractable or computationally burdensome to model explicitly~\cite{FAJEMISIN}. Such surrogates effectively capture uncertain, nonlinear, unobservable, or behaviour-driven relationships that are difficult to express using traditional formulations. Demonstrated applications include learning the physics of the distribution grid and embedding them into the optimal power flow of the distribution network~\cite{Angela}; incorporating ML-derived frequency dynamics into unit commitment models~\cite{RAJABDORRI}; and integrating a stochastic load-prediction model into virtual power plant optimisation~\cite{ALCANTARA2023120895}. By embedding these learned relationships directly into optimisation problems, OCL enables a flexible and data-driven framework to address emerging challenges in modern power system operation.

A key finding from the review in~\cite{FAJEMISIN} is that ReLU-based deep neural networks (DNNs) are among the most widely used surrogate models within OCL. Their continuous piecewise-linear structure allows them to approximate highly nonlinear functions while remaining compatible with mixed-integer programming (MIP) formulations~\cite{Chen_Ge, Arora2016-cs, GRIMSTAD}. Tight MIP reformulations, including Big-M representations enhanced with feasibility-based or optimality-based bound tightening, have been developed to embed ReLU DNNs into optimisation problems~\cite{Gurobi_ML, luegrelumip}. These formulations preserve the exactness and allow the optimiser to explicitly reason about the internal structure of the neural network.

However, a major limitation persists: each ReLU neuron introduces an additional binary variable, causing the size of the MIP to grow rapidly with network depth and width. This scalability issue often leaves the resulting optimisation problem intractable, leading most studies to restrict themselves to very shallow ReLU networks~\cite{FAJEMISIN}. The problem becomes even more pronounced when modelling high-dimensional systems or a large number of components. For example, \cite{Zhao_1} demonstrated that directly embedding a ReLU-based battery degradation model using Big-M constraints resulted in severe computational intractability in multi-battery scheduling. Similarly, \cite{ALCANTARA2023120895} reported an exponential increase in computational complexity as the size of the embedded neural network increased. These findings underscore a fundamental bottleneck: despite their expressive power, ReLU DNNs become prohibitively expensive to optimise when encoded using standard MIP reformulations.

To address this, \cite{zhang2023learning} proposed two linear-programming (LP) based relaxations, namely penalised convex area relaxation (PCAR) and penalised convex triangle area reformulation (PCTAR). These relaxations improve tractability and reduce the computational burden of optimisation with embedded DNN models. However, both rely heavily on penalty terms, which introduce sensitivity to model architecture, hyperparameter settings, and problem structure. As their sensitivity analysis showed (Section~\ref{ssec: Penalty sensitivity analysis}), the resulting formulations are not guaranteed to be tight or robust across different applications.

These limitations establish a clear research need: developing tractable yet accurate reformulations for ReLU DNNs that can be efficiently embedded into large-scale optimisation problems without relying on heuristic penalty tuning or sacrificing guaranteed tightness. Addressing this challenge is essential for unlocking the full potential of DNN-based surrogate modelling in modern power system optimisation.

During the preparation of this manuscript, a parallel study was published that adopts a similar overarching idea of integrating convexified (cvxd) DNNs into optimisation models~\cite{parallel}. The two works were developed independently, highlighting the growing interest in improving the tractability of neural-network-based surrogate models for optimisation. While the parallel study demonstrates the approach within a specific joint chance-constrained application, our work takes a broader methodological perspective. We evaluate the convexified DNN formulation across multiple neural network architectures, analyse its behaviour under varied input scenarios, and benchmark its performance against state-of-the-art (SOTA) approaches such as the PCAR and PCTAR formulations. This assessment positions our contribution as a general methodological study rather than an application-specific demonstration.
\subsection{Objectives and contributions}
\label{ssec: contributions}
\noindent In this study, we develop a method to convexify ReLU DNNs, enabling their integration as tight LP reformulations within optimisation problems. We then apply the OCL framework by learning the prosumer responsiveness function using the convexified ReLU DNN and embedding it as an LP reformulation within an aggregator’s profit maximisation problem adapted from~\cite{Yogesh_ISGT}. We further benchmark the proposed approach against state-of-the-art ReLU DNN reformulations and conventional piecewise linearisation, evaluating performance across a range of input parameters and neural network architectures.

\noindent This work differs from and extends the SOTA in several ways:
\begin{itemize}
    \item In contrast to traditional MIP-based ReLU DNN reformulations~\cite{Gurobi_ML,luegrelumip}, our approach yields an LP reformulation, improving computational efficiency.
    \item Compared to existing LP-based reformulations~\cite{zhang2023learning}, the proposed method guarantees tightness and accuracy across a broader range of network architectures and input conditions.
    \item Unlike the parallel work in~\cite{parallel}, which applies the same reformulation to a specific application without examining its broader properties, this study provides a more systematic and extensive evaluation of the methodology as a surrogate modelling strategy, including sensitivity analyses across different input ranges and ReLU DNN architectures.
\end{itemize}

We note that the proposed ReLU DNN reformulation is applicable only when the function or decision variable of interest is explicitly or implicitly minimised by the optimisation problem; further details are provided in Section~\ref{ssec: ReLU DNN}. Despite this caveat, the reformulation has strong potential for broader use in modern power system applications, including the minimisation of nonlinear generator cost functions, CO\textsubscript{2} emissions, and chance-constrained quantile functions, as demonstrated in~\cite{parallel}. It is important to emphasise that the primary objective of this paper is to demonstrate and analyse the properties of the proposed ReLU DNN reformulation as a nonlinear function surrogate and to benchmark it against SOTA surrogate modelling strategies. The optimisation model and case study presented here are designed to support this objective rather than serve as the main focus of the work.

The remainder of the paper is structured as follows: Section~\ref{sec: methodology} presents the proposed ReLU DNN reformulation and summarises the SOTA mathematical formulations. Section~\ref{sec: case study} details the case study within which the proposed formulation is implemented. Section~\ref{sec: Simulation study and results} outlines the simulation setup, presents the results, and provides a comprehensive analysis. Finally, Section~\ref{sec: Conclusions} concludes the paper and outlines directions for future work.
\section{ReLU DNN LP reformulations}
\label{sec: methodology}
In this section, we provide the mathematical background for the proposed ReLU DNN reformulation and the relevant SOTA formulations. This includes the underlying principles, notation, and structural properties required for the subsequent optimisation model.
\subsection{Proposed ReLU DNN reformulation}
\label{ssec: ReLU DNN}
In this section, we discuss the mathematical background, intuition, and key properties of the proposed ReLU DNN reformulation. 
\begin{definition}
    \label{def: ReLU func}
    For any vector $\boldsymbol{z} \in \mathbb{R}^{a}$ with any $ a\in\mathbb{N}$, we can define the element-wise ReLU function $\sigma: \mathbb{R}^{a} \rightarrow \mathbb{R}^{a}_{+}$, as shown in~\eqref{eq: ReLU definition}
    \begin{align}
        \label{eq: ReLU definition}
        \sigma\left(\boldsymbol{z}\right) = \begin{bmatrix}
        \max\left(z_{1}, 0\right) \\
        \max\left(z_{2}, 0\right) \\
        \vdots \\
        \max\left(z_{a}, 0\right)
        \end{bmatrix}
    \end{align}
\end{definition}
\begin{definition}
\label{def: ReLU DNN}
Consider a fully connected ReLU DNN $f:\mathbb{R}^{n_0} \rightarrow \mathbb{R}^{n_{L+1}}$, where $n_{0} \in \mathbb{N}$ is the input size and $L \in \mathbb{N}$ is the number of ReLU hidden layers. Let $n_{l} \in \mathbb{N}$ be the number of output neurons with weights $\boldsymbol{W_{l}} \in \mathbb{R}^{n_{l-1} \times n_{l}}$ and biases $\boldsymbol{b_{l}} \in \mathbb{R}^{n_{l}}$ for each layer $l \in \{1, 2, \dots, L, L+1\}$, where $l=L+1$ is the output layer with $n_{L+1}$ outputs. Given an input vector $\boldsymbol{z} \in \mathbb{R}^{n_{0}}$ such that $\boldsymbol{h_{0}}= \boldsymbol{z}$, the output of the hidden layers $\boldsymbol{h_{l}} \in \mathbb{R}^{n_{l}}$ and the output of the ReLU DNN are specified in~\eqref{eq: ReLU DNN hidden layer} and~\eqref{eq: ReLU DNN output}, respectively. 
\begin{subequations}
\begin{align}
    \label{eq: ReLU DNN hidden layer}
    & \boldsymbol{h_{l}} = \sigma \left[\boldsymbol{W_{l}} \, \boldsymbol{h_{l - 1}} + \boldsymbol{b_{l}} \right] & \quad \forall l \in \{1,\dots, L\} \\
    \label{eq: ReLU DNN output}
     & f(\boldsymbol{z}) = \boldsymbol{W_{L + 1}} \boldsymbol{h_{L}} + \boldsymbol{b_{L+1}} & {}
\end{align}    
\end{subequations}
\end{definition}
\noindent Based on the definitions above, we propose Theorem~\ref{theorem: A1}. 
\begin{theorem}
\label{theorem: A1}
    Let us define functions $c:\mathbb{R}^{n_0}\rightarrow \mathbb{R}$ and $g:\mathbb{R}^{n_0} \rightarrow \mathbb{R}^{e}$, where $e \in \mathbb{N}$. Let us also consider the ReLU DNN $f$ as defined in Definition~\ref{def: ReLU DNN} with a single output, i.e., $n_{L+1} = 1$. Thus, we can define a cost minimisation optimisation problem of the form below:
    \begin{subequations}
    \label{eq: standard optimisation}
    \begin{align}
        \min_{\boldsymbol{z} \in \mathbb{R}^{n_0}} & \quad c(\boldsymbol{z}) + f(\boldsymbol{z}) \\
        \text{s.t.} & \quad g\left(\boldsymbol{z}\right) \leq 0&
    \end{align}
    \end{subequations}
    For some $k \in \mathbb{N}$ and $k \leq L$, if $\boldsymbol{W_{l}} \geq 0 \text{ (element-wise)} \,\, \forall l \in \{k+1, \dots, L+1\}$, then the optimisation problem can be reformulated as
\allowdisplaybreaks
\begin{subequations}
    \label{eq: ReLU DNN reformulation}
    \begin{align}
    \min_{\boldsymbol{z} \in \mathbb{R}^{n_0}} & c(\boldsymbol{z}) + f(\boldsymbol{z}) \\
    \text{s.t.}\,& g(\boldsymbol{z}) \leq 0 \\
    &  f(\boldsymbol{z}) = \boldsymbol{W_{L+1}} \boldsymbol{h_{L}} + \boldsymbol{b_{L+1}} \\
    & \boldsymbol{h_{0}} = \boldsymbol{z} \\
    &  \boldsymbol{h_{l}} = \sigma \left[\boldsymbol{W_{l}} \, \boldsymbol{h_{l - 1}} + \boldsymbol{b_{l}} \right] & \forall l \in \{1, \dots, k-1\} \\
    & \boldsymbol{h_{l}} \geq \boldsymbol{W_{l}} \boldsymbol{h_{l - 1}} + \boldsymbol{b_{l}} & \forall l \in \{k, \dots, L\} \\
    & \boldsymbol{h_{l}} \geq 0 & \forall l \in \{k, \dots, L\}
    \end{align}
\end{subequations}
\end{theorem}
\begin{proof}
\label{proof: 1}
    Let us define a simplified version of the problem, that is, a ReLU DNN with a single neuron:
\allowdisplaybreaks
    \begin{subequations}
        \begin{align}
         \min_{z \in \mathbb{R}} & \quad f(z)& \\
         s.t. & \quad g(z) \leq 0& \\
         & \quad f(z) = \max(az, 0)& 
        \end{align}
    \end{subequations}
    The standard LP reformulation of this problem is
\allowdisplaybreaks
    \begin{subequations}
        \begin{align}
         \min_{z \in \mathbb{R}} & \quad f(z) & \\
         s.t. & \quad g(z) \leq 0 & \\
         & \quad f(z) \geq az, \,\, f(z) \geq 0 &
        \end{align}
    \end{subequations}
    
Because the problem aims to minimise $f(z)$, this reformulation will ensure $f^{*}(z)=az$ or $f^{*}(z) = 0$. If it was a maximisation problem, $f(z)$ would be unbounded. Thus, replacing $\max$ (or ReLU) constraints with linear inequalities requires minimisation of the output neuron. 
Now let us consider a problem with nested maximum constraints, such as a ReLU DNN: 
\allowdisplaybreaks
    \begin{subequations}
        \begin{align}
         \min_{\boldsymbol{z}, \boldsymbol{u} \in \mathbb{R}^{2}} & \quad f(\boldsymbol{z}) \\
         s.t. & \quad g(\boldsymbol{z, u}) \leq 0 \\
         & \quad f(\boldsymbol{z}) = \max\left(\sum_{j = 1}^{2}a_{i}z_{i}, 0\right) \\
         & \quad z_{i} = \max\left(\sum_{j = 1}^{2} b_{ij} u_{j}, 0\right) & \quad \forall i \in \{1, 2\}
        \end{align}
    \end{subequations}

Here, $a \in \mathbb{R}^{2 \times 1}$ and $ b \in \mathbb{R}^{2 \times 2}$ are the weight matrices of the ReLU DNN layers. A reformulation similar to the one demonstrated above would look as follows:
    \begin{subequations}
        \begin{align}
         \min_{\boldsymbol{z}, \boldsymbol{u} \in \mathbb{R}^{2}} & \quad f(\boldsymbol{z}) \\
         s.t. & \quad g(\boldsymbol{z, u}) \leq 0 \\
         & \quad f(\boldsymbol{z}) \geq \sum_{j = 1}^{2}a_{i}z_{i}, \, \, f(\boldsymbol{z}) \geq 0\\
         & \quad z_{i} \geq \sum_{j = 1}^{2} b_{ij} u_{j}, \,\, z_{i} \geq 0 &  \forall i \in \{ 1, 2\}&
        \end{align}
    \end{subequations}
            
If $a_{1}, a_{2}\geq 0$, the optimisation problem aims to minimise $z_{1}$ and $ z_{2}$ binding the $\max$ constraint. However, if $a_{1} < 0$ or $a_{2} < 0$, the problem aims to maximise $z_{1}$ or $z_{2}$, respectively. Thus, the associated $\max$ constraint becomes unbounded. Solving this problem will provide solutions different from the original problem. Applying this argument, we justify the reformulation presented in~\eqref{eq: ReLU DNN reformulation}. In essence, we are ``convexifying'' a section of the hidden layers (from $k$ to $L$) of a ReLU DNN by constraining their weight matrix. Owing to this convexification, the convexified section can be replaced with LP reformulations within minimisation problems, whereas the non-convexified section (from $1$ to $k-1$) cannot be linearised. We note here that the ReLU DNN can only be convexified in reverse, i.e., by constraining the weights to be non-negative from output layer onwards. 

Without loss of generality, the above proof also extends to problems in which $f(z)$ does not appear explicitly in the objective, provided that minimising the objective implicitly minimises $f(z)$; that is, minimising $c(z)$ implies that $f(z)$ is minimised over the feasible region. Therefore, a ReLU DNN can be fully replaced by an LP reformulation within an optimisation problem, with guaranteed tightness, provided that two conditions are satisfied: (1) the network output is minimised, either explicitly through the objective function or implicitly due to the optimisation structure, and (2) the weight matrices of all layers except the first ReLU layer are non-negative. 
\end{proof}
\begin{definition}
A ReLU DNN $f$, as defined in Definition~\ref{def: ReLU DNN} with a single output (i.e., $n_{L+1} = 1$), is called a cvxd ReLU DNN if all weight matrices from the second layer onwards are element-wise non-negative, i.e., $\boldsymbol{W}_{l} \geq 0$ (element-wise) for all $l \in \{2, \dots, L+1\}$.
\end{definition}
Clearly, a cvxd ReLU DNN can be equivalently represented through an LP reformulation, leveraging the results established in Theorem~\ref{theorem: A1}. Weight-constrained ReLU DNNs have been extensively studied in the literature for applications such as finance and healthcare, particularly for their monotonicity properties. However, these constraints affect the expressiveness of ReLU DNNs, limiting them to learning and accurately representing only convex and monotonic functions, as demonstrated in~\cite{constrained_nn}. Thus, cvxd ReLU DNNs can accurately learn convex and monotonically increasing functions. 

This limitation poses a caveat when applying this LP reformulation, as it trades accuracy and expressiveness for tractability and tightness, as demonstrated by our results in Section~\ref{sec: Simulation study and results}. Despite this limitation, the results in Section~\ref{sec: Simulation study and results} demonstrate that cvxd ReLU DNNs can effectively capture nonlinear and nonconvex monotonically increasing flexibility cost functions for the aggregator's operation. Consequently, cvxd ReLU DNNs may be suitable for learning nonlinear fuel cost functions, CO\textsubscript{2} emissions, or cumulative distribution functions for chance-constrained programs, which inherently increase monotonically. 
\subsection{SOTA ReLU DNN reformulations}
\label{ssec: SOTA}
\begin{figure}[htbp]
    \centering
    \begin{subfigure}[b]{0.35\linewidth}
        \centering
    \includegraphics[trim=0 280 0 290, clip, width=\linewidth]{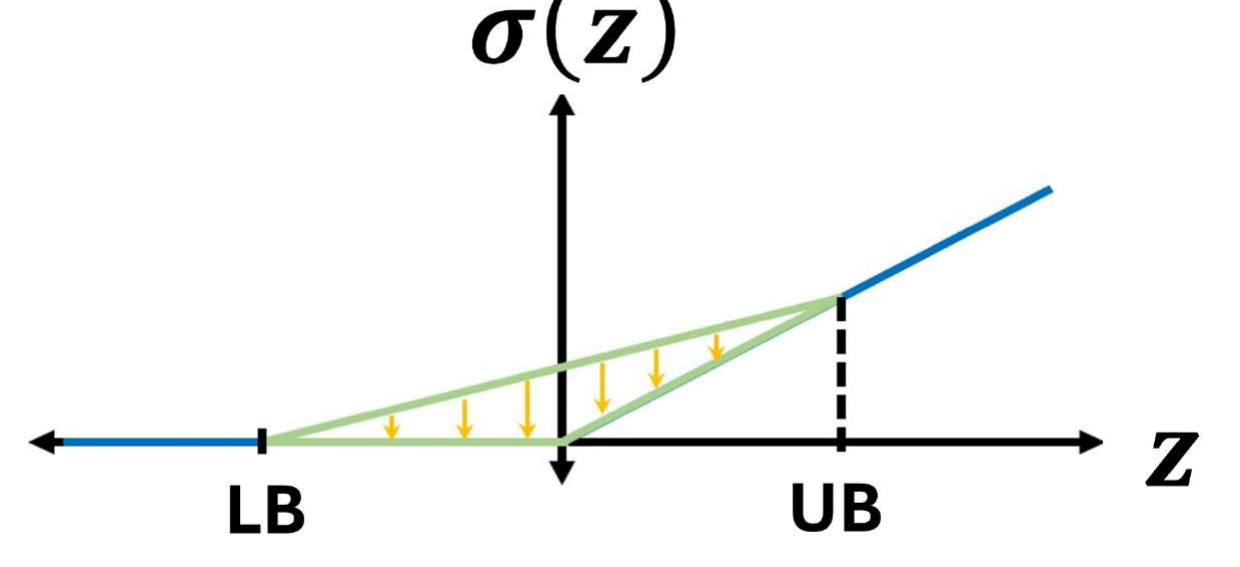}
        \caption{PCTAR reformulation}
        \label{fig: PCTAR}
    \end{subfigure}
    \hfill
    \begin{subfigure}[b]{0.35\linewidth}
        \centering
        \includegraphics[trim=0 260 0 290, clip, width=\linewidth]{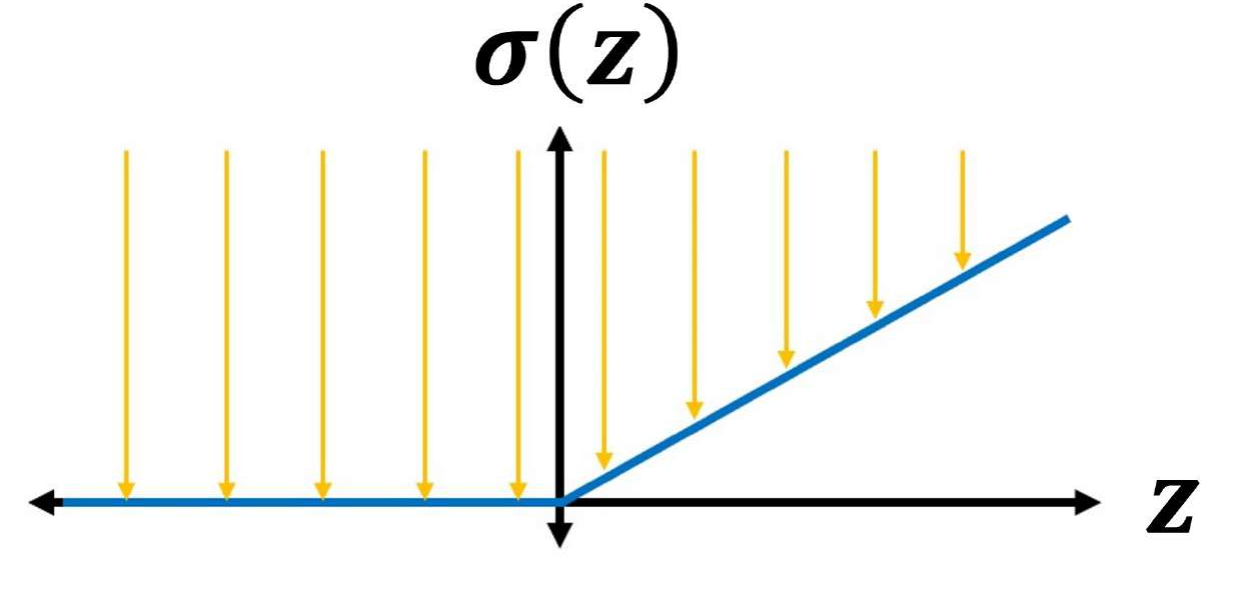}
        \caption{PCAR reformulation}
        \label{fig: PCAR}
    \end{subfigure}
    \caption{ReLU function feasible space for LP reformulations adapted from~\cite{zhang2023learning}}
    \label{fig: PCAR and PCTAR}
\end{figure}
The authors in~\cite{zhang2023learning} proposed two LP reformulations for ReLU activation functions: PCAR and PCTAR. Both methods are based on the same core idea as our proposed ReLU DNN reformulation, namely that a ReLU neuron can be linearised if its output is minimised within the optimisation problem.

In our approach, we leverage this principle by first verifying whether the optimisation objective involves minimising the output of the ReLU DNN. If so, we propagate this minimisation behaviour throughout the entire network by constraining all weight matrices of the ReLU DNN to be non-negative. This structural enforcement ensures that the outputs of all intermediate neurons are also minimised.

In contrast, PCAR and PCTAR do not impose any structural constraints on the DNN weights $\boldsymbol{W_{l}}$ for $l \in {1, \dots, L}$. Instead, they introduce penalty functions directly within the optimisation model to implicitly encourage the minimisation of neuron outputs. This makes them broadly applicable to generic ReLU DNNs with arbitrary weights. PCTAR is a further relaxation of PCAR, where the feasible space of each ReLU neuron is restricted to a triangular region, as illustrated in Fig.~\ref{fig: PCAR and PCTAR}, to improve tractability.

In the following, we briefly describe the mathematical formulations of PCAR and PCTAR under the assumption that the output of the ReLU DNN is minimised in the optimisation problem. For more general formulations, we refer the reader to the original work~\cite{zhang2023learning}. We retain the notation introduced in Section~\ref{ssec: ReLU DNN} with $n_{L+1}=1$, and define $\boldsymbol{\alpha_{l}} \in \mathbb{R}^{n_l}$ as the vector of penalty coefficients for layer $l$, and $LB, UB \in \mathbb{R}$ as the lower and upper bounds of each ReLU neuron, respectively. The formulations for PCAR and PCTAR are specified in Eqs.~\eqref{eq: PCAR obj}--\eqref{eq: PCAR end} and Eqs.~\eqref{eq: PCTAR obj}--\eqref{eq: PCTAR end}. 
\begin{subequations}
    \label{eq: PCAR reformulation}
    \begin{align}
    \label{eq: PCAR obj}
    \min_{\boldsymbol{z} \in \mathbb{R}^{n_0}} & c(\boldsymbol{z}) + f(\boldsymbol{z}) + \sum_{l=1}^{L} \boldsymbol{\alpha_{l}}^{\text{T}}\boldsymbol{h_{l}}\\
    \text{s.t.}\,& g(\boldsymbol{z}) \leq 0 \\
    &  f(\boldsymbol{z}) = \boldsymbol{W_{L+1}} \boldsymbol{h_{L}} + \boldsymbol{b_{L+1}} \\
    & \boldsymbol{h_{0}} = \boldsymbol{z} \\
    & \boldsymbol{h_{l}} \geq \boldsymbol{W_{l}} \boldsymbol{h_{l - 1}} + \boldsymbol{b_{l}} & \forall l \in \{1, \dots, L\} \\
    \label{eq: PCAR end}
    & \boldsymbol{h_{l}} \geq 0 & \forall l \in \{1, \dots, L\}
    \end{align}
\end{subequations}
\begin{subequations}
    \label{eq: PCTAR reformulation}
    \begin{align}
    \label{eq: PCTAR obj}
    \min_{\boldsymbol{z} \in \mathbb{R}^{n_0}} & c(\boldsymbol{z}) + f(\boldsymbol{z}) + \sum_{l=1}^{L} \boldsymbol{\alpha_{l}}^{\text{T}}\boldsymbol{h_{l}}\\
    \text{s.t.}\,& g(\boldsymbol{z}) \leq 0 \\
    &  f(\boldsymbol{z}) = \boldsymbol{W_{L+1}} \boldsymbol{h_{L}} + \boldsymbol{b_{L+1}} \\
    & \boldsymbol{h_{0}} = \boldsymbol{z} \\
    & \boldsymbol{h_{l}} \geq \boldsymbol{W_{l}} \boldsymbol{h_{l - 1}} + \boldsymbol{b_{l}} & \forall l \in \{1, \dots, L\} \\
    \label{eq: PCTAR triangle}
    & \boldsymbol{h_{l}} \leq k_{1}\left[\boldsymbol{W_{l}} \boldsymbol{h_{l - 1}} + \boldsymbol{b_{l}}\right] + k_{2}& \forall l \in \{1, \dots, L\} \\
    \label{eq: PCTAR end}
    & \boldsymbol{h_{l}} \geq 0 & \forall l \in \{1, \dots, L\}
    \end{align}
\end{subequations}

From the objective functions of both reformulations, Eqs.~\eqref{eq: PCAR obj} and \eqref{eq: PCTAR obj}, additional penalty terms can be observed compared with our proposed reformulation. In the case of PCTAR, an additional constraint~\eqref{eq: PCTAR triangle} is introduced to further restrict the feasible space of each ReLU neuron to a triangular region, as illustrated in Fig.~\ref{fig: PCTAR}. This constraint is defined by the parameters $k_1 = \frac{UB}{UB - LB}$ and $k_2 = \frac{UB \cdot LB}{UB - LB}$, and aims to improve the tractability by relaxing the feasible space.
\section{Case study: Energy aggregator bidding optimisation}
\label{sec: case study}
In this section, we present the nonlinear optimisation model used to evaluate the effectiveness of the proposed ReLU DNN reformulation relative to SOTA alternatives. We also outline the training strategy employed for the ReLU DNN surrogate models.
\subsection{Optimisation model}
\label{ssec: optimisation model}
The optimisation case study in this work is an energy aggregator’s optimal bidding problem for participation in day-ahead electricity scheduling markets. The formulation is adapted from the nonlinear deterministic framework in~\cite{Yogesh_ISGT}, with minor modifications; readers are referred to the original study for additional background and assumptions. A relatively simple case study is intentionally chosen to ensure tractability across all surrogate models and allow for a clear comparison amongst them.

Energy aggregators act as intermediaries between prosumers and wholesale electricity markets, coordinating flexible demand and distributed energy resources into a consolidated, tradable service. By incentivising prosumers to shift their consumption or adjust their import and export behaviour, aggregators unlock small-scale flexibility that would otherwise remain inaccessible. This improves market participation, improves the use of distributed assets, and supports overall system reliability. A central requirement in aggregator modelling is capturing prosumer price responsiveness, defined as the amount of consumption shifted in response to an incentive signal. This behaviour is highly nonlinear and nonsmooth~\cite{muf}, and embedding it directly into optimisation models can lead to computational intractability. Surrogate techniques such as piecewise linearisation (PWL)~\cite{Yogesh_ISGT} are therefore commonly used to retain tractability. In this study, we evaluate the proposed ReLU DNN reformulation as an alternative surrogate to represent this complex relationship.

Let $\mathcal{T} = \{1, \dots, T\}$ denote the set of time intervals over the horizon $T \in \mathbb{N}$. The market price profile is $\boldsymbol{\lambda}^{\text{M}} \in \mathbb{R}^{T}$ (DKK/MWh). The aggregator’s flexibility bid is $\boldsymbol{x} \in \mathbb{R}_{+}^{T}$ (MWh), bounded by the maximum flexibility available $\boldsymbol{\overline{x}} \in \mathbb{R}_{+}^{T}$. Due to rebound effects, this limit is adjusted to $\boldsymbol{\tilde{x}} \in \mathbb{R}_{+}^{T}$. The incentive offered to the prosumer group is given by $\boldsymbol{\lambda}^{\text{I}} \in \mathbb{R}_{+}^{T}$ (DKK/MWh).

We adopt the prosumer responsiveness function proposed in~\cite{Yogesh_ISGT}, shown in Eq.~\eqref{eq: Prosumer responsiveness A}, which models a saturation curve: prosumers provide little flexibility at low incentives, and their response gradually saturates to the updated maximum achievable flexibility $\tilde{x}_{t}$ at high incentives. Unlike~\cite{Yogesh_ISGT}, we replace $\boldsymbol{\overline{x}}$ with $\boldsymbol{\tilde{x}}$, which yields a more realistic representation, since rebound effects reduce future flexibility and therefore increase the incentive required to induce further shifting. As shown in Eq.~\eqref{eq: Prosumer responsiveness B}, the incentive offered to prosumers depends on the decision variables $\boldsymbol{x}$ and $\boldsymbol{\tilde{x}}$, along with the shaping parameters $q_{t}, r_{t} \in \mathbb{R}_{+}$ that capture exogenous factors for each $t \in \mathcal{T}$. 

\allowdisplaybreaks
\begin{subequations}
    \begin{align} 
        \label{eq: Prosumer responsiveness A}
        & x_{t} = \frac{\tilde{x}_{t}}{1+e^{\left(q_{t} - r_{t}\lambda^{\text{I}}_{t}\right)}} & \forall t \in \mathcal{T}
        \\
        \label{eq: Prosumer responsiveness B}
        \Rightarrow\, &\lambda^{\text{I}}_{t} = \frac{1}{r_{t}}\left[q_{t} - \ln\left(\frac{\tilde{x}_{t}}{x_{t}} - 1\right)\right] &\forall t \in \mathcal{T}
    \end{align}
\end{subequations}

The aggregator's profit maximisation problem is formulated in Eqs.~\eqref{eq: objective value}--\eqref{eq: max flexibility constraint}, following the model in~\cite{Yogesh_ISGT}. Equation~\eqref{eq: objective value} defines the objective as maximising the net profit of the aggregator. Constraint~\eqref{eq: max flexibility constraint} requires that the scheduled flexibility does not exceed the updated maximum available flexibility. The rebound effect is captured through the cross-temporal matrix $\boldsymbol{A} \in [0,1]^{T \times T}$ in Eq.~\eqref{eq: updated up max flexibility constraint}. This matrix is lower triangular with zero diagonal elements, ensuring that flexibility activation at interval $t$ reduces the available flexibility in future intervals.

\begin{subequations}
    \begin{align}
    \label{eq: objective value}
    \max_{\boldsymbol{x}, \boldsymbol{\tilde{x}}, \boldsymbol{\lambda^{\text{P}}}} \quad & \sum_{t \in \mathcal{T}} \left[\lambda_{t}^{\text{M}} x_{t} - \lambda_{t}^{\text{P}} \right] & \quad
    \\
    \label{eq: updated up max flexibility constraint}
    \text{s.t.} \quad & \tilde{x}_{t} = \overline{x}_{t} - \sum_{j \in \mathcal{T}} A_{tj} x_{t} & \, \forall t \in \mathcal{T}
    \\
    \label{eq: total purchase cost}    
    & \lambda_{t}^{\text{P}} = \lambda_{t}^{\text{I}} x_{t} = \frac{x_{t}}{r_{t}}\left[q_{t} - \ln\left(\frac{\tilde{x}_{t}}{x_{t}} - 1\right) \right]& \, \forall t \in \mathcal{T}
    \\
    \label{eq: max flexibility constraint}
    & x_{t} \leq \tilde{x}_{t} & \, \forall t \in \mathcal{T}
    \end{align}
\end{subequations}

Equation~\eqref{eq: total purchase cost} represents the total purchase cost of the flexibility procured from prosumers. This function is nonlinear and nonconvex, posing computational challenges for the profit maximisation problem. To address this, we replace it with a tractable surrogate model $\tilde{f}\left(x_{t}, \tilde{x}_{t}, q_{t}, r_{t}\right)$ that approximates the original cost. When $\tilde{f}(\cdot)$ is derived using a machine learning technique, the optimisation model falls within the OCL framework. Given that the variable of interest $\boldsymbol{\lambda^{\text{{P}}}}$ is being minimised, this problem is a good candidate for our proposed ReLU DNN reformulation. 

In this study, we consider two classes of surrogate models: PWL~\cite{PWL} and ReLU DNN-based approaches. For the latter, three reformulations are evaluated: (i) the standard MIP formulation via the Gurobi ML Python package~\cite{Gurobi_ML}, (ii) the PCAR and PCTAR LP relaxation as outlined in Section~\ref{ssec: SOTA}, and (iii) our proposed reformulation as outlined in Section~\ref{ssec: ReLU DNN}. 

PWL remains a common choice for approximating nonlinear functions but becomes computationally expensive as the input dimension or number of pieces $N_P$ increases~\cite{GRIMSTAD}. Thus, we apply it only as a baseline by simplifying Eq.~\eqref{eq: total purchase cost} to a two-dimensional form, assuming deterministic exogenous factors $q_t$ and $r_t$. In contrast, ReLU DNNs can capture multivariate nonlinear relationships more naturally; hence, the purchase cost is modelled as a four-dimensional function for DNN training.
\subsection{ReLU DNN training strategy}
\label{ssec: training strategy}
In this section, we describe the training strategy for ReLU DNNs to learn the nonlinear and nonconvex relationship presented in Eq.~\eqref{eq: total purchase cost}. We note here that despite its nonconvex nature, our aim is to demonstrate that cvxd ReLU DNNs may sufficiently capture the relationship to obtain reasonable quality solutions. The defined function has four inputs and one output and is invalid if $x_{t} \geq \tilde{x}_{t}$. It should be noted that two sets of ReLU DNNs are trained using the strategy described below: 1) \textbf{cvxd ReLU DNN}, which is implemented using the proposed formulation as specified in Section~\ref{ssec: ReLU DNN}; 2) \textbf{Unconstrained (UC) ReLU DNN}, which is implemented using the other formulations.

Let $\mathcal{D} \in \mathbb{R}_{+}^{N \times 4}$ denote the input dataset with $N_{D} \in \mathbb{N}$ data points. These points are generated randomly such that an input point $\boldsymbol{z} = \{x_{t}, \tilde{x}_{t}, q_{t}, r_{t}\}$ satisfies $\boldsymbol{\underline{z}} \leq \boldsymbol{z} \leq \boldsymbol{\overline{z}}$ and $z_{1} < z_{2}$. Here, $\boldsymbol{\underline{z}}$ and $\boldsymbol{\overline{z}}$ are the lower and upper bounds of $\boldsymbol{z}$. These are then substituted in Eq.~\eqref{eq: total purchase cost} to obtain the corresponding output dataset $\mathcal{O} \in \mathbb{R}_{+}^{N}$, so the combined dataset is denoted as $\{\mathcal{D}, \mathcal{O}\}$.

To accelerate the training process, we applied min-max normalisation to the combined dataset, normalising the input and output values to a range of $[0, 1]$. The resulting normalised combined input and output datasets, denoted as $\{\mathcal{D}, \mathcal{O}\}^{'}$, were then randomly divided into training and validation sets in an 80:20 ratio. We trained the ReLU DNNs using the Adam optimiser with a learning rate of 0.0001 while applying the mean squared error loss metric. The models were trained for 1000 epochs with a batch size of 1000. The validation set was used to monitor the training process and prevent overfitting.
\section{Simulation study and results}
\label{sec: Simulation study and results}
In this section, we describe the simulation studies implemented to evaluate the performance of the proposed reformulation in Section~\ref{ssec: ReLU DNN} compared to SOTA methods. 
\subsection{Optimisation parameters}
\label{ssec: optimisation parameters}
In this section, we describe the market and prosumer cluster parameters used within the optimisation model. We assume that the aggregator participates in the European manual frequency restoration reserve (mFRR) up-regulation capacity market, specifically in the DK1 region of Denmark, which is operated by Energinet. Up-regulation denotes increasing the amount of generation injected into the grid. This market is operated one day in advance and reserves tertiary FCAS for every hour the next day, hence $T=24$ hours for the optimisation horizon. It is compatible with the day-ahead formulation presented in Section~\ref{ssec: optimisation model}. 

\begin{figure}[htbp]
    \centering
    \begin{subfigure}[b]{0.80\linewidth}
    \centering
    \includegraphics[width=\linewidth]{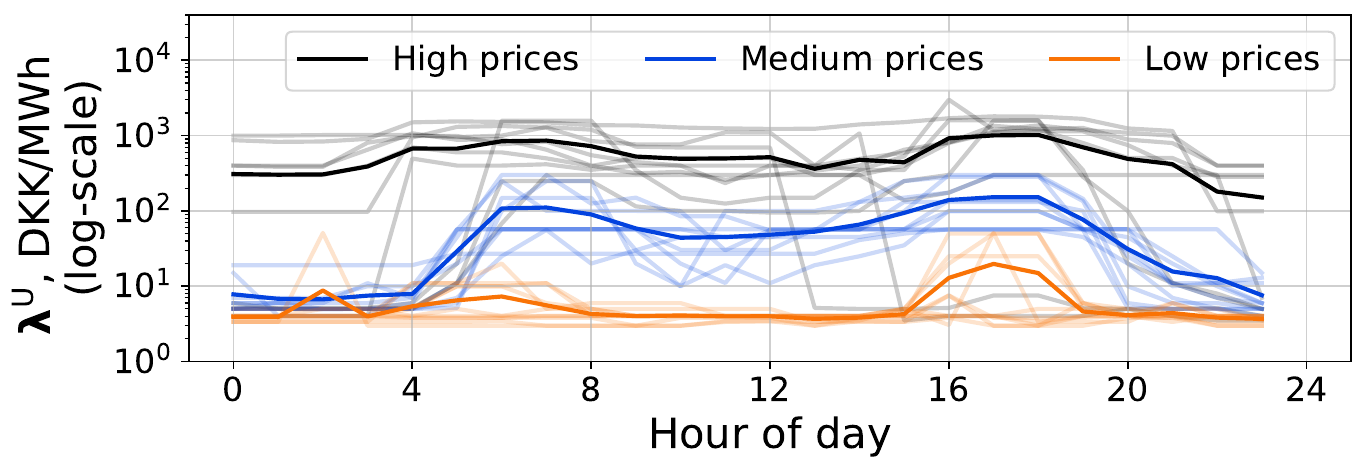}
    \subcaption{mFRR capacity market price by category}
    \label{fig: Average price scenarios}
    \end{subfigure}
    \hfill
    \begin{subfigure}[b]{0.80\linewidth}
        \centering
        \includegraphics[width=\linewidth]{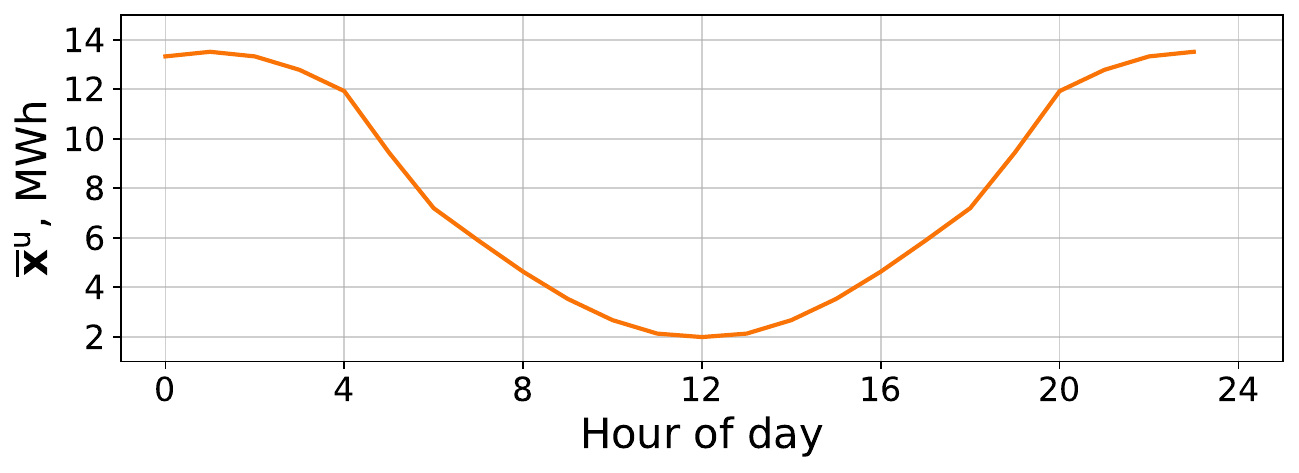}
        \subcaption{Maximum up-regulation flexibility obtainable from the prosumer cluster}
        \label{fig: max flexibility}
    \end{subfigure}
    \caption{Optimisation problem parameters}
    \label{fig: Market parameters}
\end{figure}

Real mFRR capacity market price scenarios were obtained using the Energinet data service API~\cite{energi_data_service}. The scenarios were grouped into low-, medium-, and high-price categories, and ten representative scenarios were selected from each group. This enables us to analyse how market price levels affect the tractability and performance of the surrogate models for $\boldsymbol{\lambda^{\text{P}}}$. The selected price profiles are illustrated in Fig.~\ref{fig: Market parameters}(a), with the mean profile highlighted. Notably, around 90\% of the prices are below 10~DKK/MWh, making low-price scenarios the most likely and therefore the most critical. However, when prices are high, they are orders of magnitude higher, as shown by the logarithmic scale in Fig.~\ref{fig: Market parameters}(a). The mFRR capacity market is therefore well-suited for this study due to its high price variability, which allows us to examine surrogate-model performance under different input magnitudes.

The data of the prosumer cluster was generated randomly based on the work done in~\cite{Yogesh_ISGT}. The profile of the maximum achievable flexibility ($\boldsymbol{\overline{x}}$) used in this study is presented in Fig.~\ref{fig: Market parameters}(b). Note that this profile follows a sinusoidal pattern to account for diurnal factors, as demonstrated by the work done in~\cite{Balint2019}. The cross-temporal matrices associated with bid flexibility and updated maximum flexibility ($\boldsymbol{A}$) are a randomly generated lower triangular matrix with diagonal elements equal to zero. The price elasticity of prosumers can be used as a proxy for $\boldsymbol{A}$, which can be obtained using real-world data, as shown by~\cite{MILLER2016235}. Furthermore, the total flexibility lost in future intervals must be less than or equal to the flexibility reserved in the current interval $t$, that is, $\sum_{i \in \mathcal{T}} A_{i, t} \leq 1, \forall t \in \mathcal{T}$.

For the prosumer responsiveness model, the shaping parameters $q_{t}, r_{t}$ were randomly generated and fixed for each hour $t \in \mathcal{T}$. These parameters can be derived by fitting a saturation curve to real-world data using regression~\cite{sigmoid}. For the two-dimensional PWL model~\cite{PWL}, we set the number of pieces to $N_{P} = 4$ for both inputs. 

For the PCTAR method, the bounds for the neurons were set as $LB=-10, UB=10$ to account for potential outliers, given that training with min-max normalisation typically keeps neuron inputs and outputs within $\pm 1$. To train the ReLU DNNs, $N_{D}=3\times10^{5}$ data points were generated following the methodology in Section~\ref{ssec: training strategy}. We note that the average CPU training time for all the ReLU DNN architectures used in this case study was 587 seconds. 
\subsection{Tractaility, accuracy and solution quality analysis}
\label{ssec: Tractability analysis}
In this section, we evaluate the tractability, accuracy, and solution quality achieved using different $\boldsymbol{\lambda^{\text{P}}}$ surrogate models, based on the mathematical models presented in Section~\ref{ssec: optimisation model}. For this, we trained UC and cvxd ReLU DNNs with three hidden layers consisting of 10, 20, and 10 neurons, respectively, as discussed in Section~\ref{ssec: training strategy}. These trained neural networks are used to represent $\lambda_{t}^{\text{P}}, \forall t \in \mathcal{T}$ within our optimisation model using the Gurobi ML package~\cite{Gurobi_ML}, PCTAR, PCAR, and the proposed reformulations. In addition, we employed the PWL method as a benchmark for evaluating these approaches based on neural networks. The formulated optimisation models were then solved using Gurobi\textsuperscript{\textregistered} 11.0~\cite{Gurobi} on a PC with 16GB RAM and an Intel\textsuperscript{\textregistered} i7 16-core processor. We set the stopping criteria to be 3600 seconds or MIP Gap of 1\%, whichever came first. 

Since PCAR and PCTAR are penalty-driven approaches, we tested various penalty values during optimisation, as detailed in Section~\ref{ssec: Penalty sensitivity analysis}. For comparison, we selected the penalty value that maximised the realised profit and minimised the runtime for each price scenario. 
\subsubsection{Computational performance}
\label{sssec: computational performance}
\begin{figure}[!htpb]
    \centering
    \includegraphics[width=0.80\linewidth]{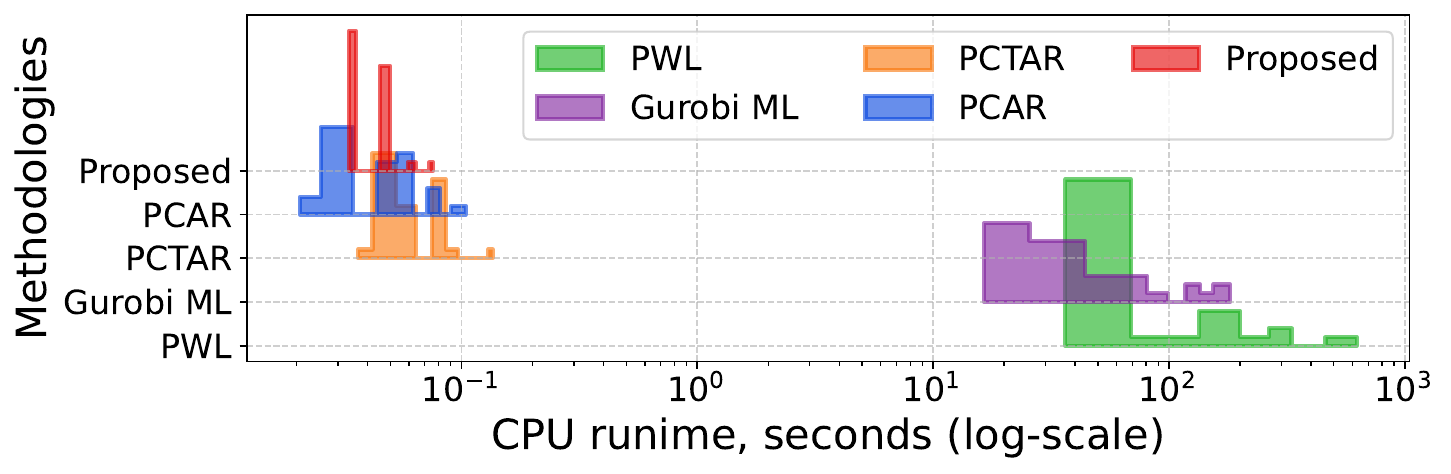}
    \label{fig: cpu runtime histogram}
    \caption{CPU runtime histogram for all surrogate models across all simulated scenarios}
    \label{fig: Tractability and accuracy results}
\end{figure}

Figure~\ref{fig: Tractability and accuracy results} shows a histogram of the CPU runtime for all price scenarios solved using various surrogate models. CPU runtimes provide insight into the amount of computer resources being used, which in turn affects the operational costs for the aggregator. Please note that the x-axis of this plot is on a log scale. As expected, we can see that the LP reformulations PCTAR, PCAR, and the proposed reformulation significantly outperform the PWL and Gurobi ML reformulations in terms of speed and CPU resource usage, by a factor ranging from hundreds to thousands. This is because parallel processing is used by Gurobi\textsuperscript{\textregistered} for branch-and-bound algorithms to solve MIPs, which requires more computational effort. 
\begin{table}[H]
\centering
\caption{Tractability analysis under different price scenarios}
\renewcommand{\arraystretch}{1.1}  
\resizebox{0.5\linewidth}{!}{%
\begin{tabular}{cccccc}
\toprule
 Price & Surrogate & Mean realised & Mean run- & Mean MIP\\
case & model & profit DKK & time & gap \%  \\
\midrule
\multirow[m]{5}{*}{High} & PWL & 50402 & 1 s & 1 \\
 & Gurobi ML & 50253 & 3 s & 1 \\
 & PCTAR & \textcolor{blue}{50861} & 44 ms & 0 \\
 & PCAR & \textcolor{blue}{50861} & \textcolor{blue}{26 ms} & 0 \\
 & Proposed & 50670 (\textcolor{red}{0.4\%}) & 34 ms & 0 \\
\cline{1-5}
\addlinespace
\multirow[m]{5}{*}{Medium} & PWL & \textcolor{blue}{3910} & 6 s & 1 \\
 & Gurobi ML & 3678 & 8 s & 1 \\
 & PCTAR & 3870 & 44 ms & 0 \\
 & PCAR & 3870 & \textcolor{blue}{31 ms} & 0 \\
 & Proposed & 3894 (\textcolor{red}{0.4\%}) & \textcolor{blue}{31 ms} & 0\\
\cline{1-5}
\addlinespace
\multirow[m]{5}{*}{Low} & PWL & \textcolor{blue}{229} & 29 s & 0 \\
 & Gurobi ML & 171 & 15 s & 1 \\
 & PCTAR & -514 & 65 ms & 0 \\
 & PCAR & -514 & 48 ms & 0 \\
 & Proposed & 226 (\textcolor{red}{1.3\%})& \textcolor{blue}{35 ms} & 0 \\
\bottomrule
\end{tabular}
}
\label{tab: Tractability analysis}
\end{table}
From Table~\ref{tab: Tractability analysis}, we note that all models are relatively fast in the high price scenarios, followed by the medium price scenarios. This can be explained in Fig.~\ref{fig: Market parameters}(a), where the high and medium price scenarios are orders of magnitude larger and more volatile than the low price cases. Thus, the decision for the optimiser is fairly straightforward, as it is easier for the branch-and-bound algorithm to reject sub-optimal solutions. Conversely, in the case of low-price scenarios, the purchase cost and mFRR prices are most likely comparable making it harder for the optimiser to identify and reject sub-optimal solutions.
\subsubsection{Solution quality and surrogate model accuracy analysis}
\label{sssec: Bidding behaviour}

\begin{figure}
        \centering
        \includegraphics[width=0.80\linewidth]{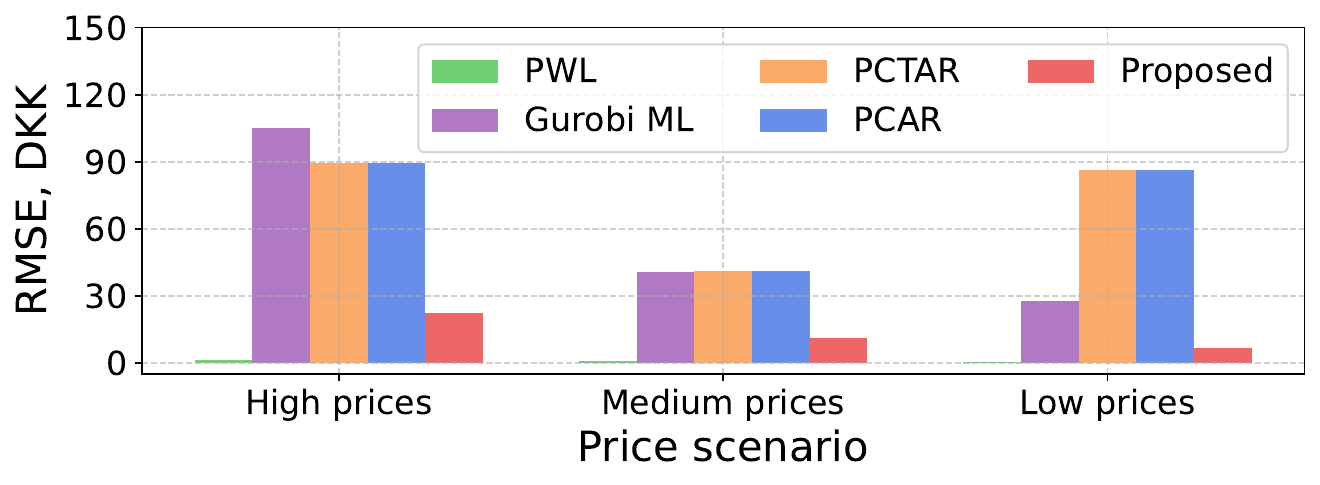}
        \caption{RMSE between actual and surrogate model $\boldsymbol{\lambda^{\text{P}}}$ across different price scenarios}
        \label{fig: rmse price scenarios}
\end{figure}

\begin{figure}[htpb]
    \centering
    \begin{subfigure}[b]{0.80\linewidth}
    \centering
    \includegraphics[width=\linewidth]{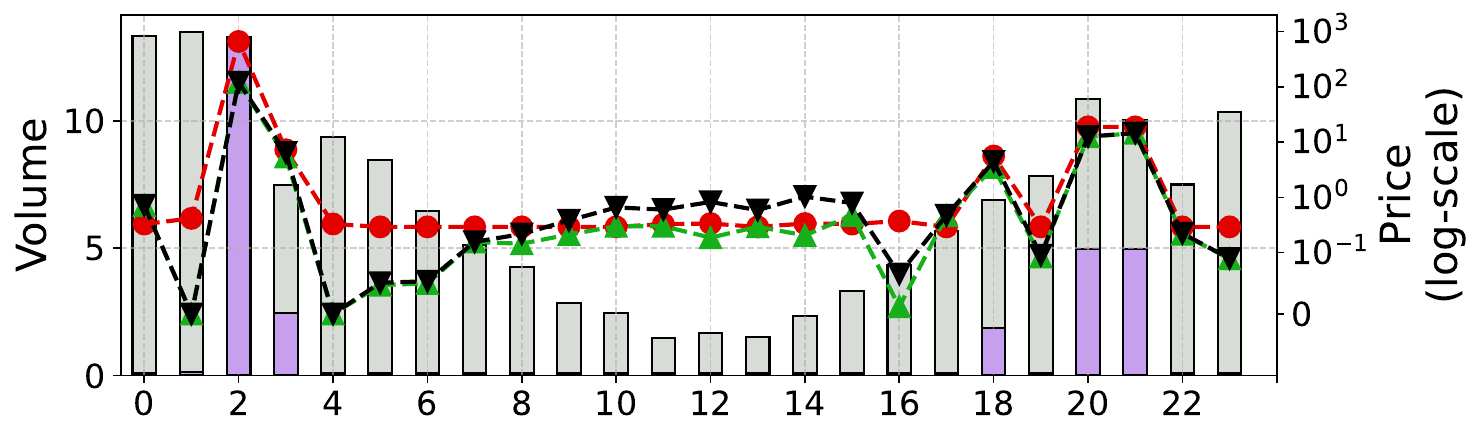}
    \subcaption{PWL}
    \label{fig: bidding pwl}
    \end{subfigure}
    \hfill
    \begin{subfigure}[b]{0.80\linewidth}
        \centering
        \includegraphics[width=\linewidth]{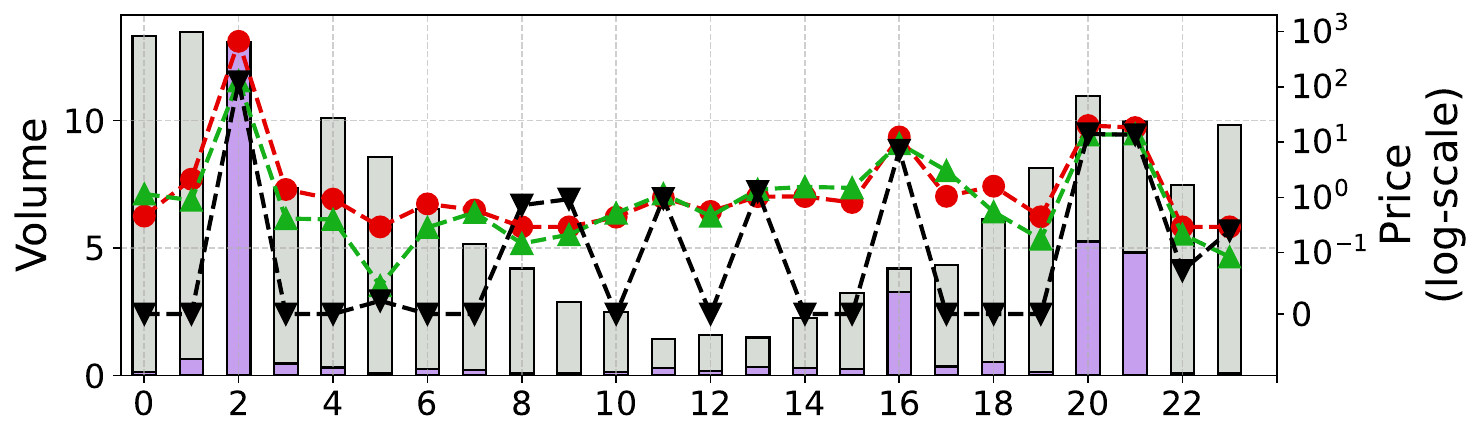}
        \subcaption{Gurobi ML}
        \label{fig: bidding gurobiml}
    \end{subfigure}
    \hfill
        \begin{subfigure}[b]{0.80\linewidth}
        \centering
        \includegraphics[width=\linewidth]{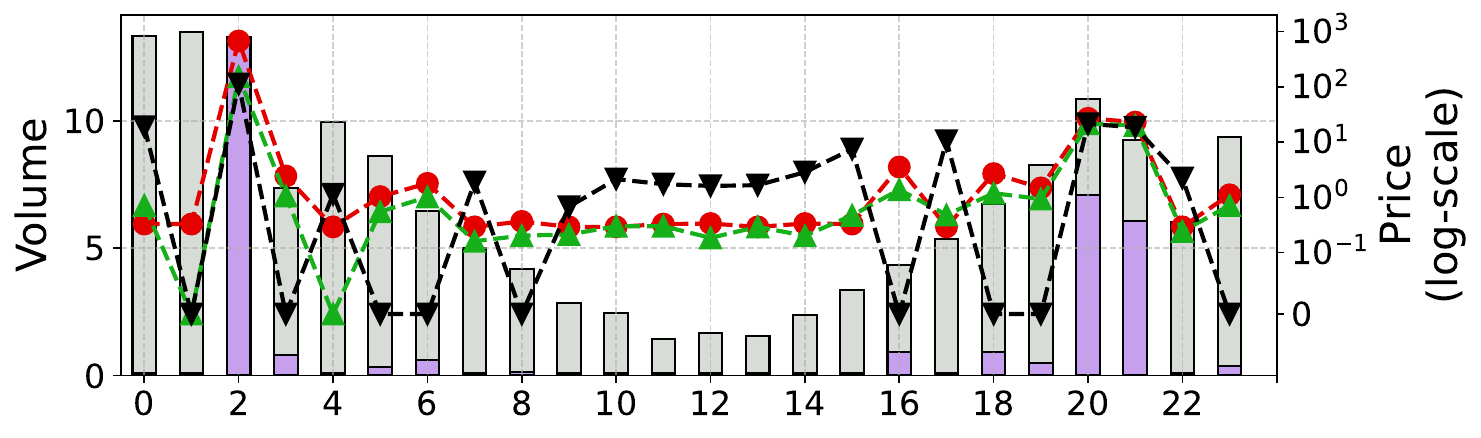}
        \subcaption{Proposed}
        \label{fig: bidding nwlp}
    \end{subfigure}
    \hfill
        \begin{subfigure}[b]{0.80\linewidth}
        \centering
        \includegraphics[width=\linewidth]{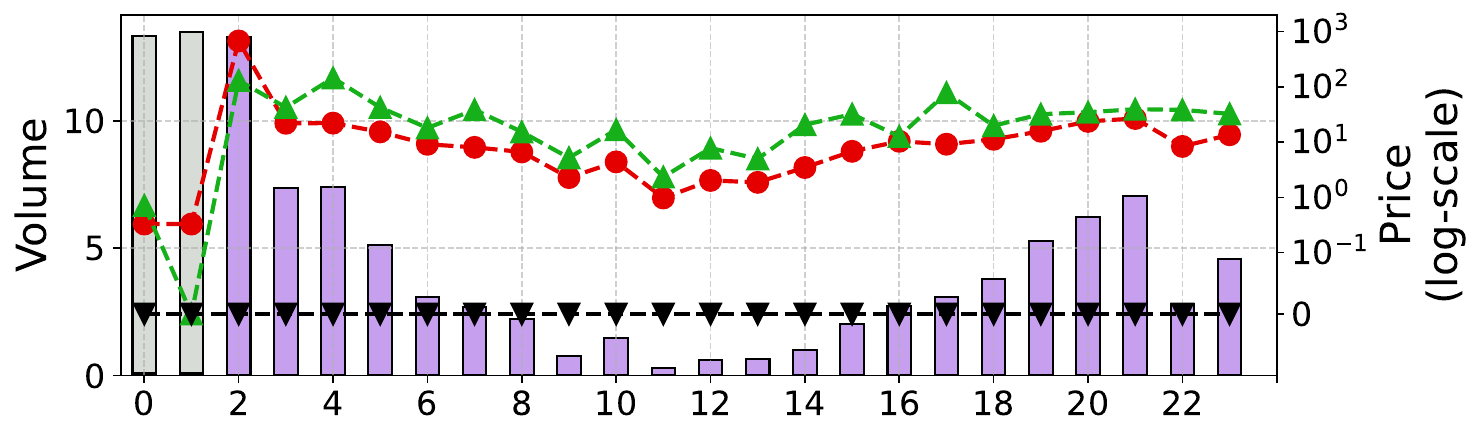}
        \subcaption{PCTAR}
        \label{fig: bidding pctar}
    \end{subfigure}
    \hfill
        \begin{subfigure}[b]{0.80\linewidth}
        \centering
        \includegraphics[width=\linewidth]{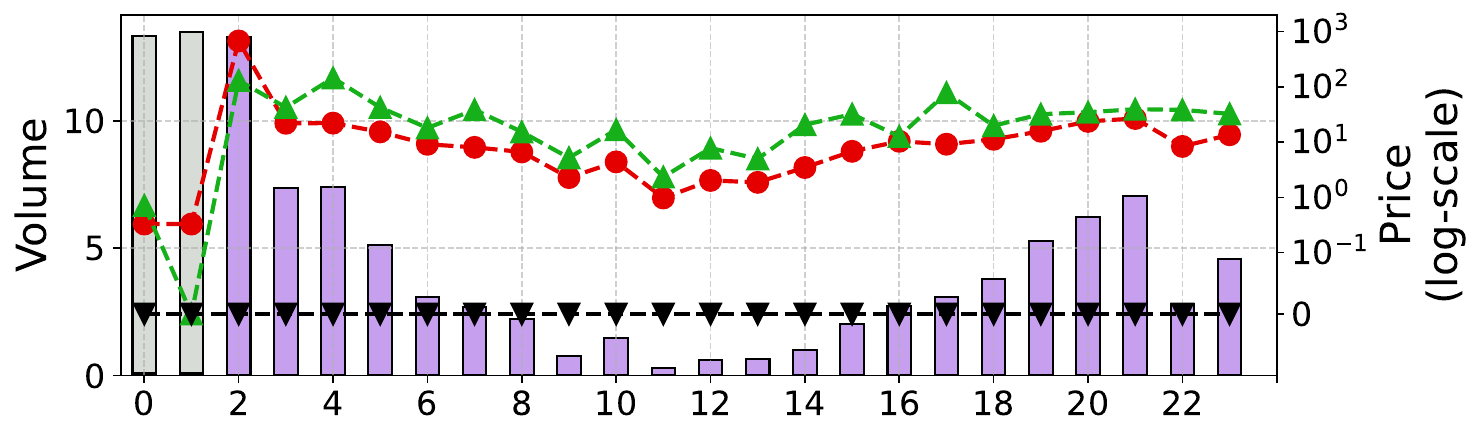}
        \subcaption{PCAR}
        \label{fig: bidding pcar}
        \hfill
    \end{subfigure}
        \begin{subfigure}[b]{\linewidth}
    \centering
    \includegraphics[width=\linewidth, trim=0 0 0 0, clip]{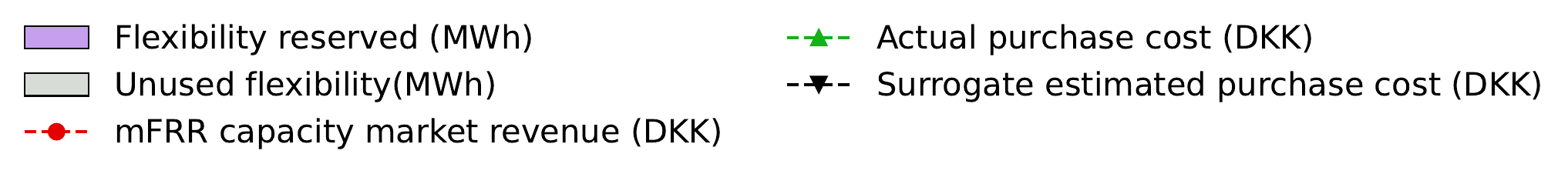}
    \label{fig: shared legend}
    \end{subfigure}
    \caption{Optimisation results for one low price instance}
    \label{fig: Bidding behaviour}
\end{figure}

In Table~\ref{tab: Tractability analysis}, the realised profit is calculated by substituting the optimal decision vectors $\boldsymbol{x}$ and $\boldsymbol{\tilde{x}}$ into Eqs.~\eqref{eq: objective value} and ~\eqref{eq: total purchase cost}, reflecting the true profit based on the actual post-optimisation purchase cost of flexibility. PCTAR and PCAR yield the highest average profit in high price scenarios, while PWL performs best under medium and low prices. As shown in Fig.~\ref{fig: rmse price scenarios}, PWL also delivers the most accurate cost estimates, which is expected given its ability to closely approximate nonlinear functions with enough segments. This is further substantiated by Fig.~\ref{fig: Bidding behaviour}(a), where the estimated cost of the up-regulation flexibility is very close to the actual cost. However, as discussed in Section~\ref{sssec: computational performance}, the PWL method is significantly more computationally expensive than our proposed method, and there also exist practical limitations, as discussed in Section~\ref{ssec: key remarks}, which may make it unsuitable in the real world.

Gurobi ML employs an MIP-based reformulation of the ReLU DNN, which is expected to provide a tight embedding of the network within the optimisation model. Notably, it uses the UC ReLU DNN to learn the $\boldsymbol{\lambda^{\textbf{P}}}$ relationship, which achieved a lower training root mean square error (RMSE) compared to the cvxd ReLU DNN used in our proposed formulation, as demonstrated using Figs.~\ref{fig: rmse length and width}(a) and (b). Therefore, it was reasonable to expect Gurobi ML to exhibit superior realised RMSE. However, it deviated from this expected trend, underperforming our model in terms of realised accuracy. This deviation was also observed across other network architectures, as discussed in Section~\ref{ssec: Length and width analysis}. The exact reason remains unclear and may be attributed to the proprietary nature of Gurobi\textsuperscript{\textregistered} solution algorithms.

PCAR and PCTAR often lead to losses in low price scenarios, both the most probable and critical, due to substantial mis-estimation of the purchase cost of the up-regulation flexibility. This is reflected in the high RMSE values in Fig.~\ref{fig: rmse price scenarios}. While these approaches perform adequately in high- and medium-price settings, they struggle when market prices are low. Figures~\ref{fig: Bidding behaviour}(d) and (e) illustrate this issue for a representative low-price scenario. Despite bidding maximum flexibility at most intervals, PCAR and PCTAR estimate the purchase cost as $0$ DKK, which leads to actual costs exceeding revenues and resulting in losses.
This behaviour stems from the inability of penalty functions to consistently enforce the minimisation of all the neurons output. In many cases, the optimiser can decouple intermediate neurons from the ReLU DNN feed-forward equations, setting them freely to force low estimated costs while bidding maximum flexibility. As a result, the trained ReLU DNN's behaviour cannot always be reliably replicated, an issue further discussed in Section~\ref{ssec: Length and width analysis}.

Our proposed model consistently achieves the second highest realised profit in all scenarios, remaining within approximately 1\% of the top-performing method, as indicated by the percentage differences in parentheses in Table~\ref{tab: Tractability analysis}. It also ranks second in realised RMSE as shown in Fig.~\ref{fig: rmse price scenarios}, outperforming Gurobi ML as noted above. Although our formulation employs cvxd ReLU DNNs to learn $\boldsymbol{\lambda^{\textbf{P}}}$, architectures that are typically more suited to capture convex nonlinearities still perform well in this nonconvex setting, providing solid accuracy and solution quality at a fraction of computational speed.

It is interesting to note from Figs.~\ref{fig: Bidding behaviour}(b) and (c) that both Gurobi ML and the proposed method are able to reasonably estimate the up-regulation flexibility cost when the reserved up-regulation flexibility volume is relatively high. However, for lower volumes, both methods tend to underestimate or overestimate the cost. This discrepancy may stem from the underlying ReLU DNNs struggling to accurately learn the cost relationship at lower volume levels. Consequently, the training performance of the embedded ReLU DNNs plays a significant role in the ultimate bidding behaviour and effectiveness of these approaches. This highlights the broader challenge of integrating learned models within optimisation frameworks, especially under data regimes where performance generalisation is critical.
\subsection{ReLU DNN length and width sensitivity analysis}
\label{ssec: Length and width analysis}
In this section, we evaluate the tractability, accuracy and tightness of the ReLU DNN-based approximation methods across varying network architectures. Two sets of cases are considered: one in which the number of neurons (i.e., the width) is varied while keeping the number of hidden layers (length) fixed, and another where the length is varied with a fixed total width. We focus exclusively on the ten low price scenarios, as they represent the most likely and critical operating conditions.

As described in Section~\ref{ssec: Tractability analysis}, for each architecture, we train two sets of ReLU DNNs: the UC ReLU DNN embedded using Gurobi ML, PCTAR, and PCAR reformulations, and the cvxd ReLU DNN embedded using our proposed reformulation. Additionally, we perform a penalty sensitivity analysis for both PCAR and PCTAR across all architectures, as detailed in Section~\ref{ssec: Penalty sensitivity analysis}. Recall that the baseline ReLU DNN architecture used in Section~\ref{ssec: Tractability analysis} consists of three hidden layers with 10, 20, and 10 neurons, respectively, resulting in a total width of 40.

For the width sensitivity analysis, we fix the length at three layers and vary the total width to 20, 80, 200, and 400 neurons. The neurons are distributed across the layers in the same proportion as in the baseline model.

For the length sensitivity analysis, we fix the total width to 40 neurons (total width of the baseline ReLU DNN) and vary the number of layers using the following architectures: 1) single layer with 40 neurons; 2) two layers with 20 neurons each; 3)four layers with 5, 15, 15, 5 neurons; and 4) five layers with 2, 8, 20, 8, 2 neurons.  

\subsubsection{Tractability and solution quality analysis}
\label{sssec: width length tractability}
\begin{table}[htpb]
\centering

\begin{minipage}{0.49\textwidth}
\centering
\caption{Width sensitivity analysis}
\renewcommand{\arraystretch}{1}
\resizebox{\linewidth}{!}{%
\begin{tabular}{ccccc}
\toprule
 Width & Surrogate & Mean realised & Mean run- & Mean MIP\\
{} & model & profit DKK & time & gap \%  \\
\midrule
\multirow[m]{4}{*}{20} & Gurobi ML & 227 & 2 s & 0 \\
 & PCAR & -214 & 15 ms & 0 \\
 & PCTAR & -214 & \textcolor{blue}{8 ms} & 0 \\
 & Proposed & \textcolor{blue}{231} & 16 ms & 0 \\
\cline{1-5}
\addlinespace
\multirow[m]{4}{*}{80} & Gurobi ML & 221 & 3600 s & 127 \\
 & PCAR & -36 & 78 ms & 0 \\
 & PCTAR & -36 & 181 ms & 0 \\
 & Proposed & \textcolor{blue}{237} & \textcolor{blue}{56 ms} & 0 \\
\cline{1-5}
\addlinespace
\multirow[m]{4}{*}{200} & Gurobi ML & 115 & 3600 s & 884 \\
 & PCAR & -147 & 549 ms & 0 \\
 & PCTAR & -147 & 921 ms & 0 \\
 & Proposed & \textcolor{blue}{239} & \textcolor{blue}{259 ms} & 0 \\
\cline{1-5}
\addlinespace
\multirow[m]{4}{*}{400} & Gurobi ML & -18 & 3601 s & 236\hypertarget{targetpoint}{*} \\
 & PCAR & -9 & 3 s & 0 \\
 & PCTAR & -22 & 6 s & 0 \\
 & Proposed & \textcolor{blue}{238} & \textcolor{blue}{748 ms} & 0 \\
\bottomrule
\end{tabular}
}
\label{tab: Width tractability analysis}
\end{minipage}
\hfill
\begin{minipage}{0.49\textwidth}
\centering
\caption{Length sensitivity analysis}
\renewcommand{\arraystretch}{1}
\resizebox{\linewidth}{!}{%
\begin{tabular}{ccccc}
\toprule
 Length & Surrogate & Mean realised & Mean run- & Mean MIP\\
{} & model & profit DKK & time & gap \%  \\
\midrule
\multirow[m]{4}{*}{1} & Gurobi ML & 234 & 4 s & 1 \\
 & PCAR & -359 & \textcolor{blue}{13 ms} & 0 \\
 & PCTAR & -359 & 27 ms & 0 \\
 & Proposed & \textcolor{blue}{239} & 18 ms & 0 \\
\cline{1-5}
\addlinespace
\multirow[m]{4}{*}{2} & Gurobi ML & 233 & 87 s & 1 \\
 & PCAR & -118 & 38 ms & 0 \\
 & PCTAR & -118 & 70 ms & 0 \\
 & Proposed & \textcolor{blue}{237} & \textcolor{blue}{30 ms} & 0 \\
\cline{1-5}
\addlinespace
\multirow[m]{4}{*}{4} & Gurobi ML & 209 & 53 s & 1 \\
 & PCAR & 106 & 37 ms & 0\\
 & PCTAR & 106 & 59 ms & 0 \\
 & Proposed & \textcolor{blue}{234} & \textcolor{blue}{27 ms} & 0 \\
\cline{1-5}
\addlinespace
\multirow[m]{4}{*}{5} & Gurobi ML & 156 & 3600 s & 55 \\
 & PCAR & 47 & \textcolor{blue}{12 ms} & 0 \\
 & PCTAR & 47 & 32 ms & 0 \\
 & Proposed & \textcolor{blue}{219} & 35 ms & 0 \\
\bottomrule
\end{tabular}
}
\label{tab: Length tractability analysis}
\end{minipage}
\end{table}
Tables~\ref{tab: Width tractability analysis} and \ref{tab: Length tractability analysis} present the results of the width and length sensitivity analyses, respectively. As expected, across all ReLU DNN architectures, the LP-based surrogate models (PCAR, PCTAR, and the proposed method) consistently outperform Gurobi ML's MIP-based formulation in terms of computational efficiency. Notably, the width of the ReLU DNN has a substantially greater impact on computational time than its length. This is because increasing the number of neurons directly increases the number of constraints in all surrogate models, thus increasing computational complexity. For Gurobi ML in particular, the addition of neurons leads to more integer variables, which significantly increases runtime due to the exponential growth in the complexity of solving MIPs.

In contrast, the length of the ReLU DNN appears to have minimal impact on the computational performance of the LP-based surrogate models. However, for Gurobi ML, the mean runtime generally increases with network length, except in the case of length 4, where a reduction is observed. For length 5, the mean runtime abruptly spikes to 3600 seconds, suggesting potential scalability issues with deeper networks. Since the exact workings of the Gurobi ML's algorithm are not publicly disclosed, it is difficult to pinpoint the exact cause of this irregularity.

In both Tables~\ref{tab: Width tractability analysis} and \ref{tab: Length tractability analysis}, the mean realised profit is re-evaluated post-optimisation, as described in Section~\ref{sssec: Bidding behaviour}. The results clearly demonstrate that our proposed formulation consistently yields the highest profit across all ReLU DNN architectures. In contrast, the PCAR and PCTAR models generally incur losses due to the mis-estimation of the up-regulation flexibility cost, as discussed in Section~\ref{sssec: Bidding behaviour}. The only instances where these models make a profit occur when the number of hidden layers exceeds three, suggesting that the length of the ReLU DNN has some influence on the quality of their solution. However, even in these cases, the realised profit is less than half of what is achieved by our proposed method.

Meanwhile, Gurobi ML shows relatively better performance than PCAR and PCTAR in terms of realised profit, but its performance deteriorates significantly with increasing network width and length. This is primarily due to its poor scalability, as the solver struggles to find high-quality feasible solutions within the 3600-second time limit. This is evident from the increasing mean MIP gap percentages. Notably, for the architecture with a width of 400 neurons, Gurobi ML was unable to solve 8 out of 10 price scenarios within the time limit, and the reported mean MIP gap is based only on the two solvable cases, as indicated by \hyperlink{targetpoint}{*}. These results indicate that Gurobi ML does not tractably scale to wider ReLU DNNs, which poses challenges for applications requiring larger neural network architectures.

Overall, these findings demonstrate that our proposed ReLU DNN reformulation significantly outperforms existing SOTA methods in both computational efficiency and solution quality, across all tested ReLU DNN architectures.
\subsubsection{Surrogate model accuracy and tightness analysis}
\label{sssec: width length accuracy}
\begin{figure}[htbp]
    \centering
    \begin{subfigure}[b]{0.80\linewidth}
    \centering
    \includegraphics[width=\linewidth]{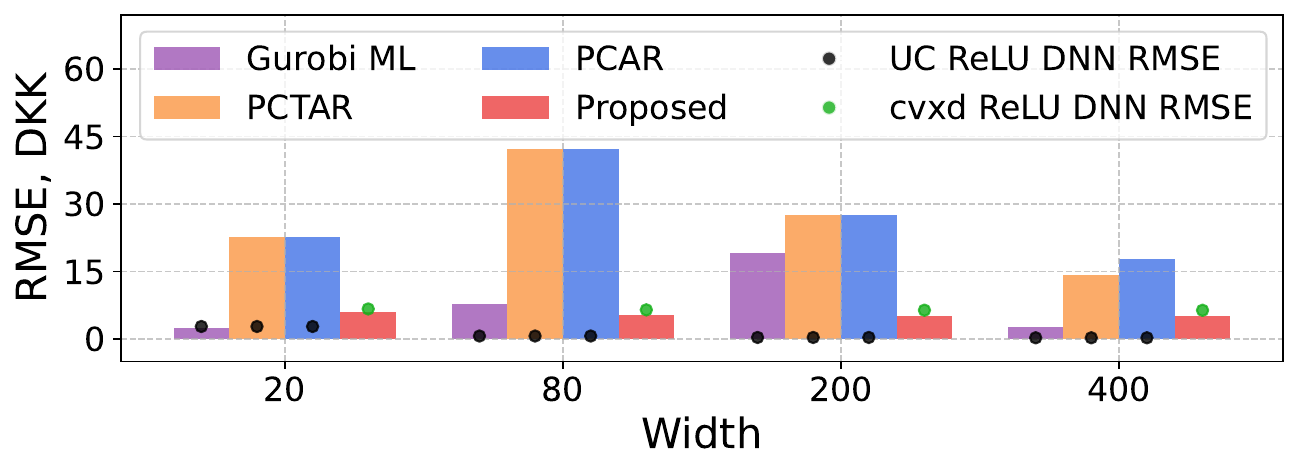}
    \subcaption{Varying ReLU DNN widths}
    \label{fig: rmse width scenarios}
    \end{subfigure}
    \hfill
    \begin{subfigure}[b]{0.80\linewidth}
        \centering
        \includegraphics[width=\linewidth]{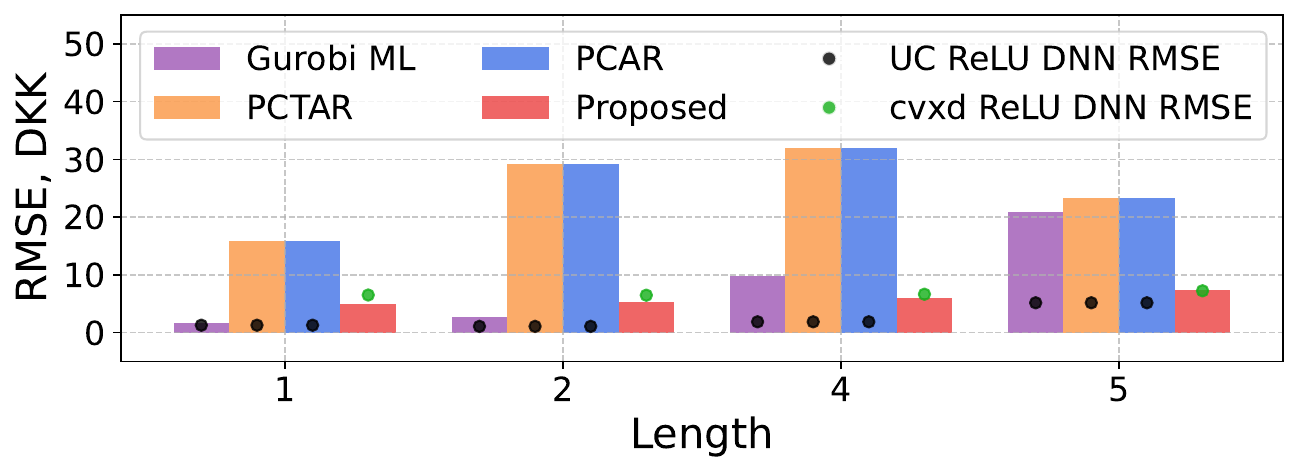}
        \subcaption{Varying ReLU DNN lengths}
        \label{fig: rmse length scenarios}
    \end{subfigure}
    \caption{RMSE between actual and surrogate model $\boldsymbol{\lambda^{\text{P}}}$ across varying ReLU DNN architectures, along with their corresponding training RMSE}
    \label{fig: rmse length and width}
\end{figure}
Figures~\ref{fig: rmse length and width}(a) and (b) present the RMSE between the true up-regulation flexibility price, computed post-optimisation, and the surrogate price generated within the optimisation model, across various ReLU DNN architectures. In addition, the training RMSE is reported for each architecture.  It is important to note that for each ReLU DNN architecture, we simulated the ten low price scenario over a 24-hour period, resulting in 240 data points per architecture, which were used to compute these RMSE values. 

The results demonstrate that our proposed ReLU DNN reformulation consistently achieves a lower RMSE across all tested architectures. In contrast, both PCAR and PCTAR exhibit the highest RMSE, indicating the poorest accuracy among the compared methods. Its accuracy also varies inconsistently across the ReLU DNN architectures, making it hard to comment on its behaviour with increasing length and width. The Gurobi ML model performs well for smaller and shallower ReLU DNNs (e.g., width of 20 and depth below 3), but its accuracy deteriorates as network complexity increases. An exception is observed at a width of 400, where Gurobi ML achieves low RMSE; however, as noted in Section~\ref{sssec: width length tractability}, only two of ten cases were feasibly solved in this configuration, which may skew the RMSE results. These variations in the accuracy of the surrogate model directly impact the mean realised profit, as discussed in Section~\ref{sssec: width length tractability}, where a higher accuracy corresponds to greater profit in most cases.

Further supporting the robustness of our approach, we observe that the proposed model consistently achieves RMSE values close to the training RMSE of the cvxd ReLU DNN. This behaviour clearly demonstrates the tightness of our formulation, as established in Theorem~\ref{theorem: A1}. For shallower and smaller DNNs, the Gurobi ML model achieves RMSE values close to the training RMSE of the UC ReLU DNN. However, as the network becomes deeper and more complex, this tightness is lost, and the resulting RMSEs become significantly higher than the corresponding UC ReLU DNN training RMSE. This behaviour is unexpected, as Gurobi ML is assumed to be based on an MIP reformulation of ReLU DNNs, which should, in principle, maintain tightness~\cite{Gurobi_ML}.

As expected, both PCAR and PCTAR yield RMSE values that are substantially higher than those of the UC ReLU DNNs they approximate, clearly indicating that they are not tight ReLU DNN reformulations. The reasoning behind this behaviour is detailed in Section~\ref{sssec: Bidding behaviour}, and stems from the structural uncoupling of ReLU DNN neurons from feed-forward equations, leading to reduced fidelity. These observations reinforce the conclusion that, unlike other SOTA techniques, our proposed model is not highly sensitive to the underlying ReLU DNN architecture and consistently maintains both accuracy and tightness across scenarios. 

Finally, we examine the effect of varying the ReLU DNN architectures on the training RMSEs themselves. It is evident that the cvxd ReLU DNN produces relatively consistent RMSEs across all architectures, which can be attributed to its limited expressivity due to the non-negativity constraints on the hidden layer weight matrices.  In contrast, UC ReLU DNNs show better performance as the network size increases, benefiting from their greater expressiveness. They are also highly accurate, achieving significantly lower RMSEs compared to the cvxd ReLU DNNs, demonstrating their strong learning capabilities. This highlights a core limitation of our proposed formulation: learning-based optimisation methods are inherently bounded by the quality of the learning models they embed. Thus, if the embedded model (e.g., the cvxd ReLU DNN) is suboptimal due to its training constraints, the overall performance of the OCL framework may suffer despite tight reformulations.

\subsection{Penalty sensitivity analysis for PCAR and PCTAR}
\label{ssec: Penalty sensitivity analysis}
\begin{figure}[t]
    \centering
    \begin{subfigure}[b]{0.70\linewidth}
    \centering
    \includegraphics[width=\linewidth]{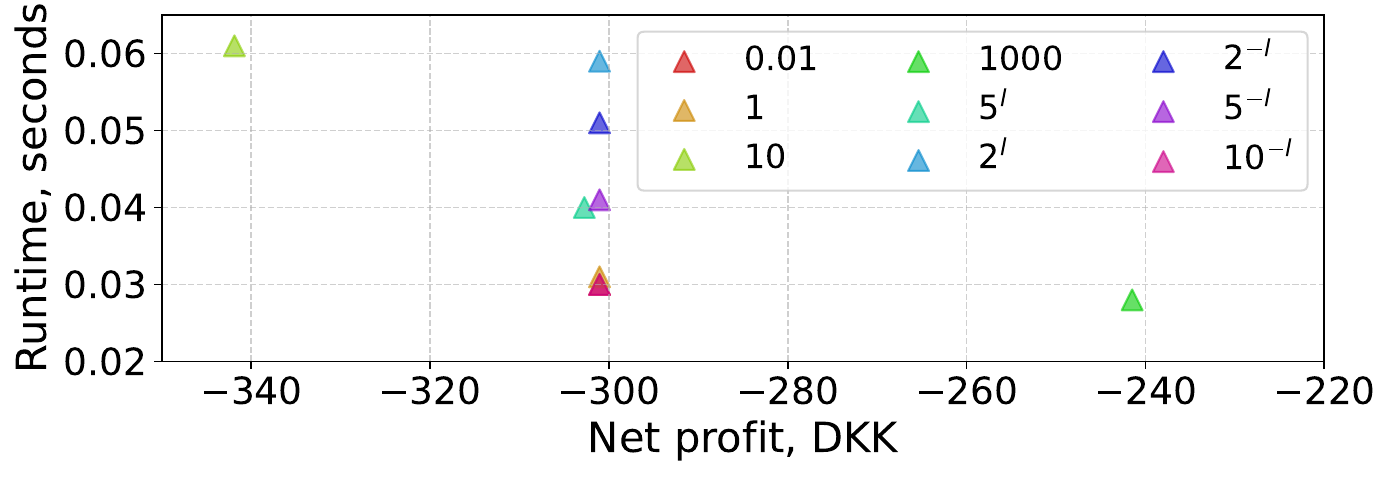}
    \subcaption{Low price scenario}
    \label{fig: example A}
    \end{subfigure}
    \hfill
    \begin{subfigure}[b]{0.70\linewidth}
        \centering
    \includegraphics[width=\linewidth]{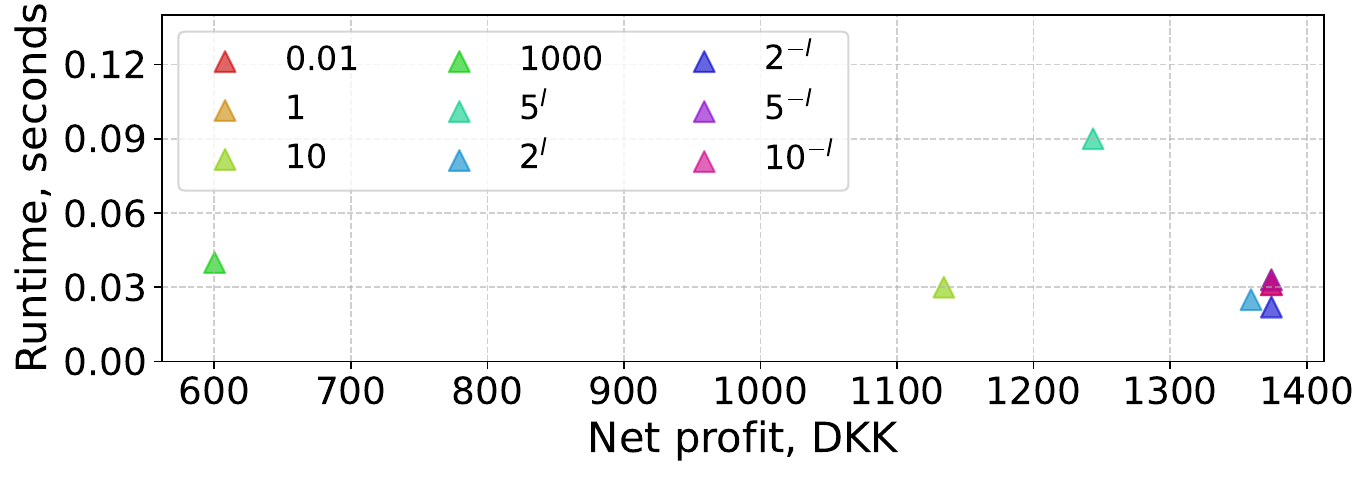}
        \subcaption{Medium price sceanrio}
        \label{fig: example B}
    \end{subfigure}
        \hfill
    \begin{subfigure}[b]{0.70\linewidth}
        \centering
    \includegraphics[width=\linewidth]{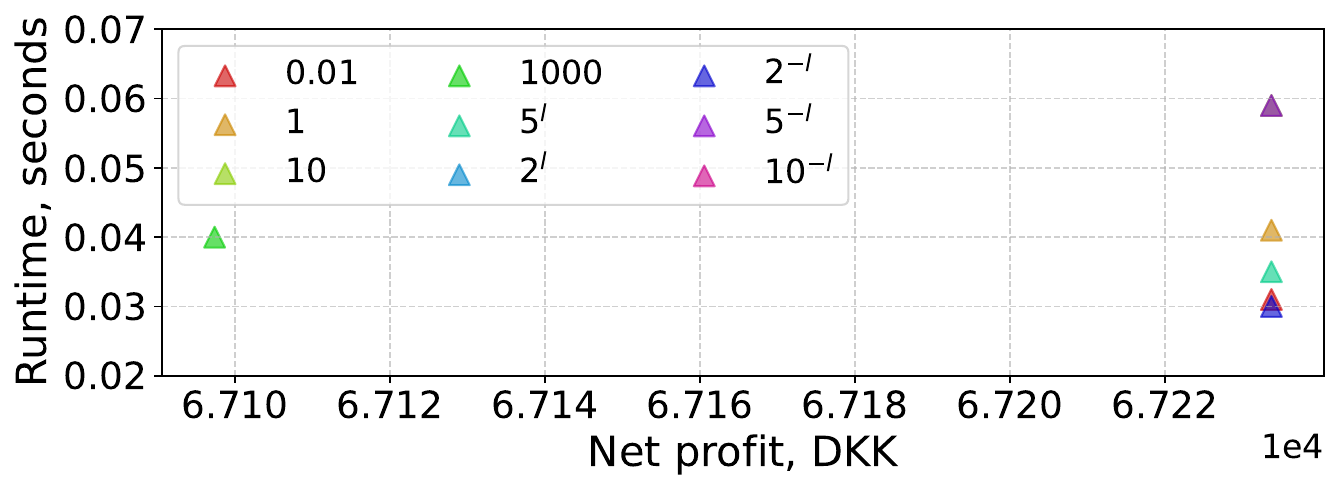}
        \subcaption{High price scenario}
        \label{fig: example B}
    \end{subfigure}
    \caption{Sample penalty sensitivity analysis examples for each price scenario}
    \label{fig: Penalty sensitivity examples}
\end{figure}
Because PCTAR and PCAR are penalty-driven approaches, we performed a penalty sensitivity analysis to determine the best penalty value for each price scenario and the ReLU DNN architecture. We considered two classes of penalty values: 1) constant penalty values applied to all neurons $0.01, 1, 10, 1000$ and 2) layer-based penalty values $5^{l}, 2^{l}, 2^{-l}, 5^{-l}, 10^{-l}$ for all neurons in the $l$\textsuperscript{th} hidden ReLU layer. 

Figure~\ref{fig: Penalty sensitivity examples} shows sample penalty sensitivity results for the three sets of price scenarios. The x-axis specifies the realised profit and the y-axis specifies the runtime for the scenario. The best penalty value is one which maximises realised profit and minimises runtime; hence, values closer to the bottom right corner of the graph. From the three examples, we can see that the best penalty value may change across cases and the profit obtained varies significantly with the penalty value chosen. 

Furthermore, we identified the optimal penalty value across all cases for the high price, medium price, and low price scenarios: $5^{-l}, 2^{-l}, 1000$ for the PCTAR reformulation and $1, 5^{-l}, 1000$ for the PCAR reformulation, based on the study conducted in Section~\ref{ssec: Tractability analysis}. Regardless of the variations in length and width discussed in Section~\ref{ssec: Length and width analysis}, a penalty value of 1000 consistently yielded the best results for all architectures. This suggests a strong correlation between the penalty value and the underlying optimisation parameters. This may work well in this application where most of the price scenarios fall into the low price case. However, in other applications where these parameters are more evenly distributed, identifying the best penalty value may prove more challenging.

Across all cases, the realised RMSE of PCAR and PCTAR are generally comparable, except when the width is set to 400, where PCAR achieves a lower RMSE. This discrepancy arises from insufficient $LB, UB$ selection in the PCTAR case, which did not fully capture the neuron outputs, leading to a worse realised profit value, as shown in Table~\ref{tab: Width tractability analysis}. This highlights the additional complexity involved in tuning the parameters for the PCTAR method. 

In summary, this section demonstrates that the performance of PCAR and PCTAR is highly sensitive to the choice of penalty values. Achieving competitive accuracy and runtime performance requires careful calibration, and improper tuning can lead to significant deterioration in both surrogate model fidelity and economic outcomes. 
\subsection{Discussion and practical considerations}
\label{ssec: key remarks}
In this section, we discuss the key insights obtained from this study and highlight practical considerations when implementing the various surrogate models.

The PWL approach is widely used in power system optimisation and showed the strongest performance in our study, accurately capturing prosumer behaviour and delivering the most favourable economic outcomes (Section~\ref{ssec: Tractability analysis}). It also solves the problem within about 30 seconds, which is acceptable for day-ahead scheduling. However, our case study was intentionally simplified to maintain tractability for both the PWL and Gurobi ML formulations. In practice, prosumer behaviour depends on many interdependent factors that cannot be captured through closed-form expressions and instead require data-driven modelling, which significantly increases computational effort. As a result, the PWL approach becomes impractical for large-scale or real-time applications, particularly as decarbonisation technologies expand. Although our analysis focuses on an aggregator setting, the insights extend to many decarbonisation-related optimisation problems that face similar modelling and tractability challenges.

Our results also highlight the potential of ReLU DNNs to effectively model nonlinear and nonconvex prosumer behaviour, even with relatively simple architectures (Section~\ref{sssec: width length accuracy}). This includes cvxd ReLU DNNs, which are theoretically well-suited for approximating convex functions. This suggests the need for further exploration, particularly as more data on prosumer behaviour becomes available. However, tractable reformulations are required to embed such models into optimisation problems. Although the standard MIP reformulation, coupled with bound-tightening strategies (as in Gurobi ML), offers theoretical tightness guarantees and performs well for shallow networks, our findings in Section~\ref{ssec: Length and width analysis} reveal that deeper or more complex ReLU DNNs pose computational challenges and can even compromise tightness, resulting in worse economic outcomes. In such cases, LP-based reformulations, such as those proposed in~\cite{zhang2023learning} and in this work, may provide more viable alternatives under the right conditions.

Although the PCAR and PCTAR methods are presented as generic LP reformulations for ReLU DNNs applicable to a wide range of optimisation problems, they inherently trade off tightness and fidelity due to their penalty-driven structure. As a result, despite their generality, they tend to yield negative economic outcomes in our case study (Sections~\ref{ssec: Tractability analysis} and~\ref{ssec: Length and width analysis}). A central challenge lies in tuning the penalty parameters, which are highly sensitive to both the scale of optimisation variables and the structure of the ReLU DNN (Sections~\ref{ssec: Length and width analysis} and~\ref{ssec: Penalty sensitivity analysis}). This sensitivity makes PCAR and PCTAR difficult to implement reliably in practice, as they offer no formal tightness guarantees for any fixed penalty setting—even if acceptable performance is observed in specific instances. Consequently, even if the original ReLU DNN captures complex system dynamics accurately, these surrogate reformulations may fail to reflect them within the optimisation problem.

By contrast, our proposed cvxd ReLU DNN effectively captures the nonlinear and nonconvex prosumer response function defined in this case study. When embedded in the aggregator’s profit maximisation problem using our LP reformulation, it achieves good quality solutions comparable to the PWL method at a fraction of the computational cost (Section~\ref{ssec: Tractability analysis}). Moreover, it demonstrates consistent performance across a range of optimisation scales and ReLU DNN architectures. As such, cvxd ReLU DNNs, when paired with our proposed LP formulation, offer a promising pathway to improving computational efficiency while preserving solution quality within the OCL framework. That said, their applicability is limited by both the restricted expressiveness of cvxd ReLU DNNs and the ReLU DNN output minimisation structure required by our reformulation, confining them to specific problem classes. Nonetheless, they warrant further exploration as a viable surrogate modelling strategy to learn other similar nonlinear and nonconvex functions, such as CO\textsubscript{2} emissions or nonlinear generator cost functions, within modern power system optimisation problems.

Finally, our findings suggest that the learning behaviour of the ReLU DNN within the OCL framework plays a critical role in determining overall optimisation performance. As illustrated in Fig.~\ref{fig: Bidding behaviour}(b)–(e) and discussed in Sections~\ref{sssec: Bidding behaviour} and~\ref{sssec: width length accuracy}, the quality of the learned approximation significantly influences the effectiveness of the surrogate model when integrated into the bidding strategy of the aggregator. While this observation may appear intuitive, it underscores a practical challenge: ensuring that the ReLU DNN is trained with sufficient fidelity and generalisation capability to accurately reflect underlying system dynamics. In real-world implementations, this requires careful consideration of training data quality, network architecture, and regularisation techniques, all of which are essential for realising the full benefits of the OCL framework.
\section{Conclusions and future scope}
\label{sec: Conclusions}
\noindent This paper developed and evaluated a scalable surrogate modelling strategy for embedding ReLU DNNs within power system optimisation. We proposed an LP reformulation for convexified ReLU DNNs, where weights beyond the first layer are constrained to be non-negative, enabling an exact and tight representation when the network output is minimised. Within an energy aggregator’s bidding optimisation problem for the Danish mFRR market, the convexified ReLU DNN approach achieves comparable economic performance to PWL, while providing substantially reduced computational burden and consistent fidelity across neural network architectures. 

\noindent The proposed method avoids the sensitivity and penalty tuning required in PCAR/PCTAR, and remains tractable where the Gurobi-ML MIP reformulation becomes computationally prohibitive. While applicability is limited to optimisation models that minimise the ReLU DNN output and to functions learnable by convexified networks, the results demonstrate the viability of the approach for surrogate modelling of complex flexibility behaviour. This opens promising directions for surrogate modelling in power system optimisation problems, including applications in generator cost modelling, CO\textsubscript{2} emissions estimation, and chance-constrained formulations. 
\section*{CRediT authorship contribution statement}
\noindent\textbf{Yogesh Pipada Sunil Kumar:} Conceptualisation, Methodology, Software, Formal analysis, Investigation, Writing - original draft, Writing - review and editing. \textbf{S. Ali Pourmousavi:} Conceptualisation, Supervision, Project administration, Funding acquisition, Writing - review and editing. \textbf{Jon A.R. Liisberg:} Supervision, Writing - review and editing. \textbf{Julian Lesmos-Vinasco:} Supervision, Writing - review and editing.
\section*{Acknowledgement}
\noindent This project is supported by the Australian Government Research Training Program (RTP) through the University of Adelaide, and a supplementary scholarship provided by Watts A/S, Denmark. During the preparation of this work, the authors used ChatGPT~\cite{chatgpt} in order to improve readability and language. After using this tool, the authors reviewed and edited the content as needed and take full responsibility for the content of the publication.
\bibliography{bibliography}
\bibliographystyle{elsarticle-num}

\end{document}